\documentclass[12pt]{article}
\usepackage{epsfig, epsf, graphicx, subfigure}
\usepackage{pstricks, pst-node, psfrag}
\usepackage{amssymb,amsmath,bm,amsthm}
\usepackage{verbatim,enumerate}
\usepackage{rotating, lscape}
\usepackage{setspace}
\def\boxit#1{\vbox{\hrule\hbox{\vrule\kern6pt
          \vbox{\kern6pt#1\kern6pt}\kern6pt\vrule}\hrule}}
\def\refhg{\hangindent=20pt\hangafter=1}
\def\refmark{\par\vskip 2mm\noindent\refhg}

\def\refhg{\hangindent=20pt\hangafter=1}
\def\refmark{\par\vskip 2mm\noindent\refhg}

\def\bse{\begin{eqnarray*}}
\def\ese{\end{eqnarray*}}
\def\be{\begin{eqnarray}}
\def\ee{\end{eqnarray}}
\def\bq{\begin{equation}}
\def\eq{\end{equation}}
\def\bse{\begin{eqnarray*}}
\def\ese{\end{eqnarray*}}

\newcommand{\bY}{\mathbf{Y}}
\newcommand{\bZ}{\mathbf{Z}}

\newcommand{\btheta}{\boldsymbol{\theta}}
\newcommand{\0}{\mathbf{0}}

\newtheorem{theo}{Theorem}

\newtheorem{lem}{Lemma}
\begin{document}

\thispagestyle{empty} \baselineskip=28pt \vskip 5mm
\begin{center} {\Huge{\bf Efficient Maximum Approximated Likelihood Inference for Tukey's $g$-and-$h$ Distribution}}
\end{center}

\baselineskip=12pt \vskip 10mm

\begin{center}\large
Ganggang Xu\footnote[1]{
\baselineskip=10pt Department of Mathematical Sciences,
Binghamton University, State University of New York,
Binghamton, New York 13902, USA. \\
E-mail: gang@math.binghamton.edu}
and Marc G.~Genton\footnote[2]{
\baselineskip=10pt CEMSE Division,
King Abdullah University of Science and Technology,
Thuwal 23955-6900, Saudi Arabia.\\
E-mail: marc.genton@kaust.edu.sa}
\end{center}

\baselineskip=17pt \vskip 10mm 

\begin{center}
{\large{\bf Abstract}}
\end{center}
Tukey's $g$-and-$h$ distribution has been a powerful tool for data exploration and modeling since its introduction. However, two long standing challenges associated with this distribution family have remained unsolved until this day: how to find an optimal estimation procedure and how to make valid statistical inference on unknown parameters. To overcome these two challenges, a computationally efficient estimation procedure based on maximizing an approximated likelihood function of the Tukey's $g$-and-$h$ distribution is proposed and is shown to have the same estimation efficiency as the maximum likelihood estimator under mild conditions. The asymptotic distribution of the proposed estimator is derived and a series of approximated likelihood ratio test statistics are developed to conduct hypothesis tests involving two shape parameters of Tukey's $g$-and-$h$ distribution. Simulation examples and an analysis of air pollution data are used to demonstrate the effectiveness of the proposed estimation and testing procedures.

\baselineskip=14pt

\par\vfill\noindent
{\bf Some key words:} Approximated likelihood ratio test; Computationally efficient; Maximum approximated likelihood estimator; Skewness; Tukey's $g$-and-$h$ distribution.
\par\medskip\noindent
{\bf Short title}: Tukey's $g$-and-$h$ distribution

\clearpage\pagebreak\newpage \pagenumbering{arabic}
\baselineskip=24.8pt

\section{Introduction}\label{sec:intro}

Datasets with skewed and/or heavy-tailed distributions are typical in many research areas. There have been numerous attempts to search for flexible and practically useful distribution families to model such data in the statistical community; see Jones (2015) for a comprehensive review. An attractive class of distributions introduced by Tukey (1977), and later named Tukey's $g$-and-$h$ distribution, has been extensively studied by many researchers; see, for example, Martinez \& Iglewicz (1984), Hoaglin (1985), Morgenthaler \& Tukey (2000). Let $Z$ be a random variable from a standard normal distribution, $N(0,1)$. A random variable, $Y$, is said to have a Tukey's $g$-and-$h$ distribution if it is obtained through the transformation
\begin{equation}
Y=\xi+\omega \tau_{g,h}(Z), \label{TGHdef}
\end{equation}
where $\xi \in \Bbb{R}$ is a location parameter, $\omega>0$ is a scale parameter, and
\begin{equation}
\tau_{g,h}(z)=g^{-1}\{\exp(gz)-1\}\exp(hz^2/2) \label{taugh}
\end{equation}
 is a one-to-one monotone function of $z \in \Bbb{R}$ for $h\geq 0$, $g \in \Bbb{R}$. When $g=0$, we use the customary definition of $\tau_{0,h}(z)=\lim_{g\to 0}\tau_{g,h}(z)=z\exp(hz^2/2)$. To simplify, from now on, values of all quantities involving the parameter $g$ evaluated at $g=0$ are defined as their limits attained at $g\to 0$. Aside from $\xi$ and $\omega$, two additional shape parameters,  $g$ and $h$, are introduced to accommodate the potential existence of skewness and heavy-tailness in the data distribution. More precisely, $g > 0$ yields a right-skewed distribution while $g<0$ corresponds to a left-skewed distribution. In the special case of $g=h=0$, the resulting distribution reduces to a normal distribution with mean $\xi$ and variance $\omega^2$. By setting $h=0$, one obtains a shifted log-normal distribution,
whereas letting $g=0$ gives a Pareto-like distribution. In fact, it has been shown that many commonly used distributions can be well-approximated by Tukey's $g$-and-$h$ distribution (Martinez \& Iglewicz, 1984; MacGillivray, 1992; Jim\'{e}nez \& Arunachalam, 2011). In Figure~\ref{fig-5}, we present three examples of using Tukey's $g$-and-$h$ distributions to approximate other distributions, where all parameters were estimated by the proposed estimation procedure using $10,000$ random numbers generated from each distribution. As we can see, all three approximations appear to be quite good.
\begin{figure}[t!]
\centering
\begin{tabular}{ccc}
\includegraphics[width=0.32\textwidth]{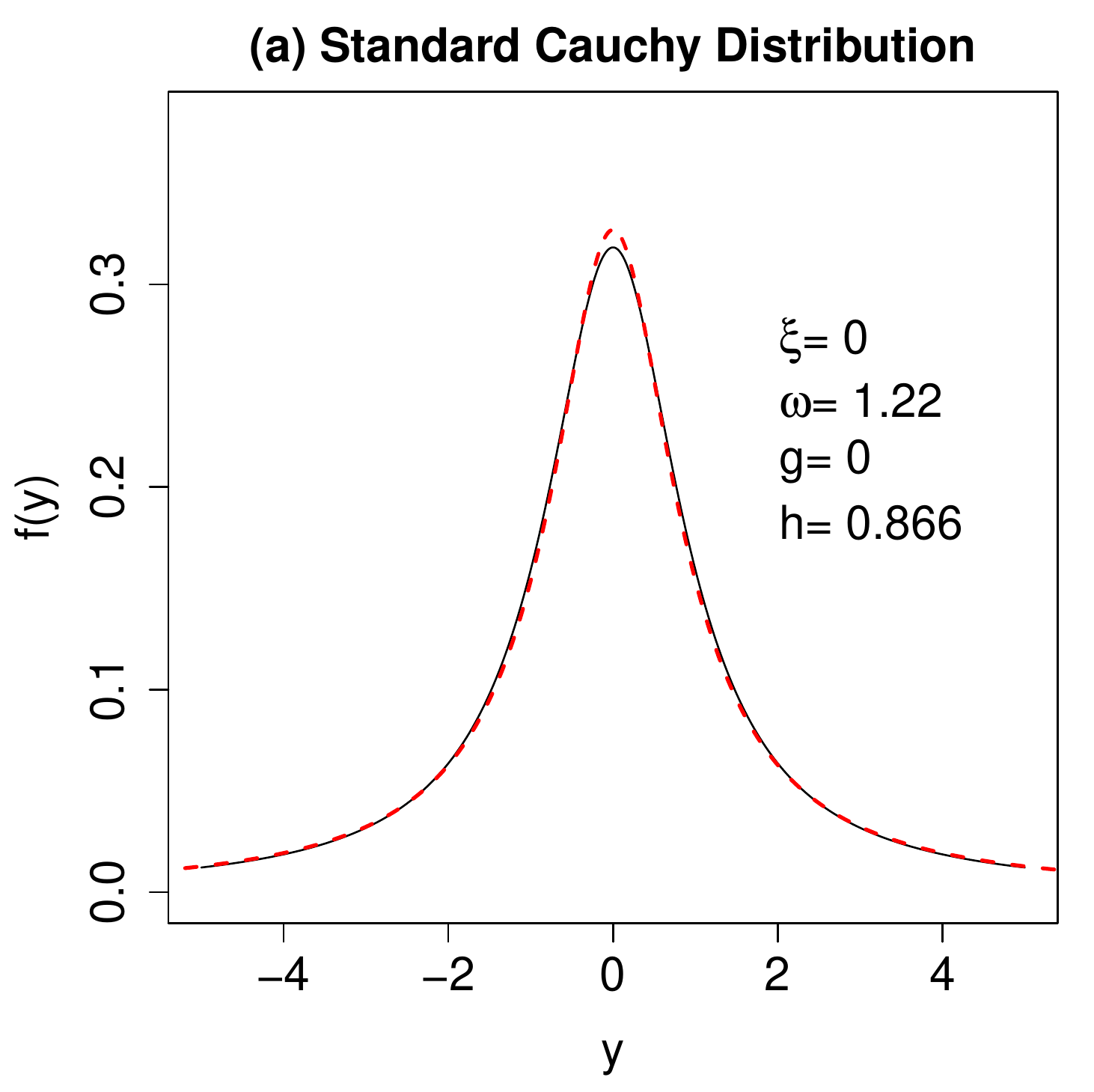} &
\includegraphics[width=0.32\textwidth]{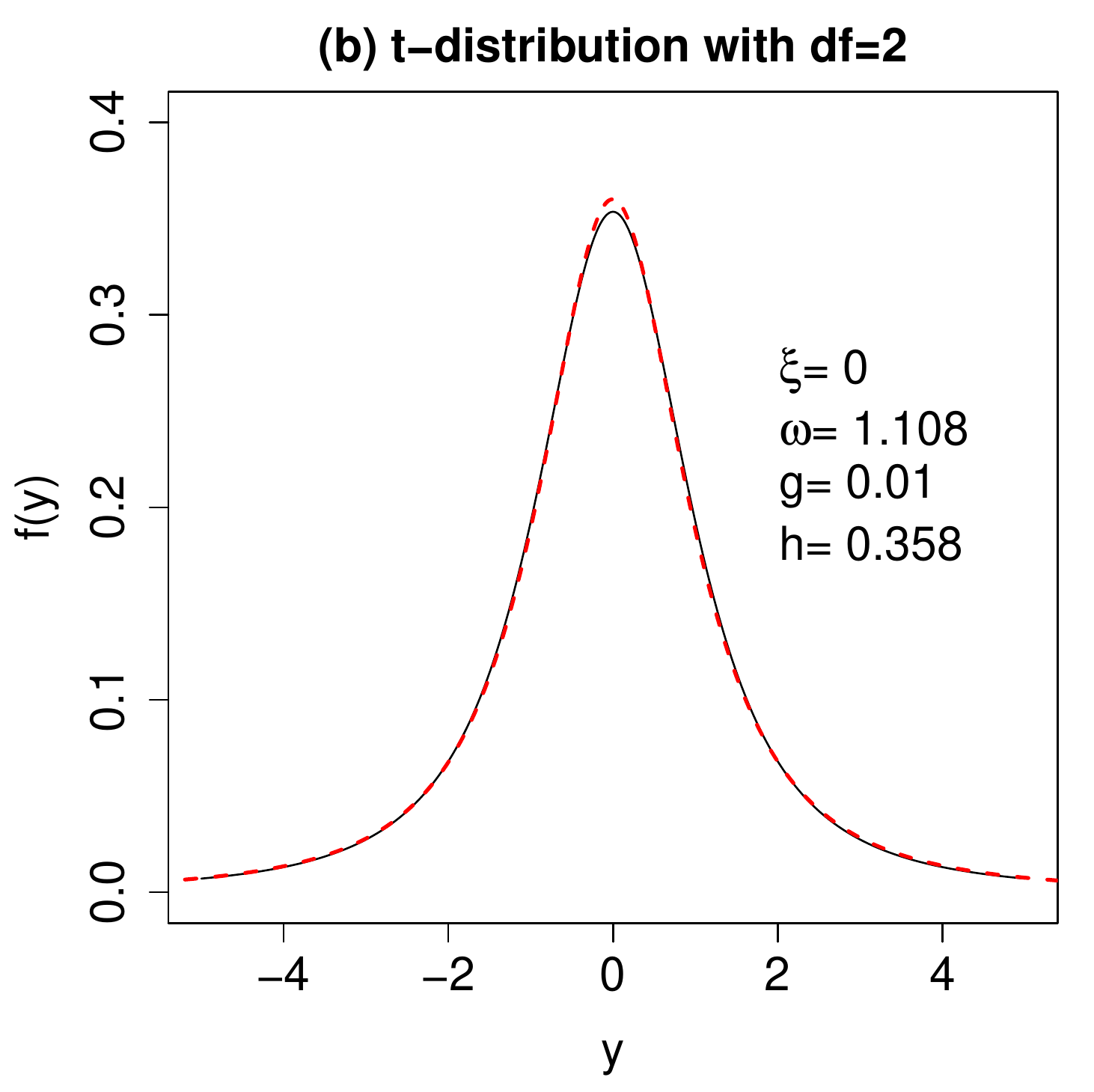} &
\includegraphics[width=0.32\textwidth]{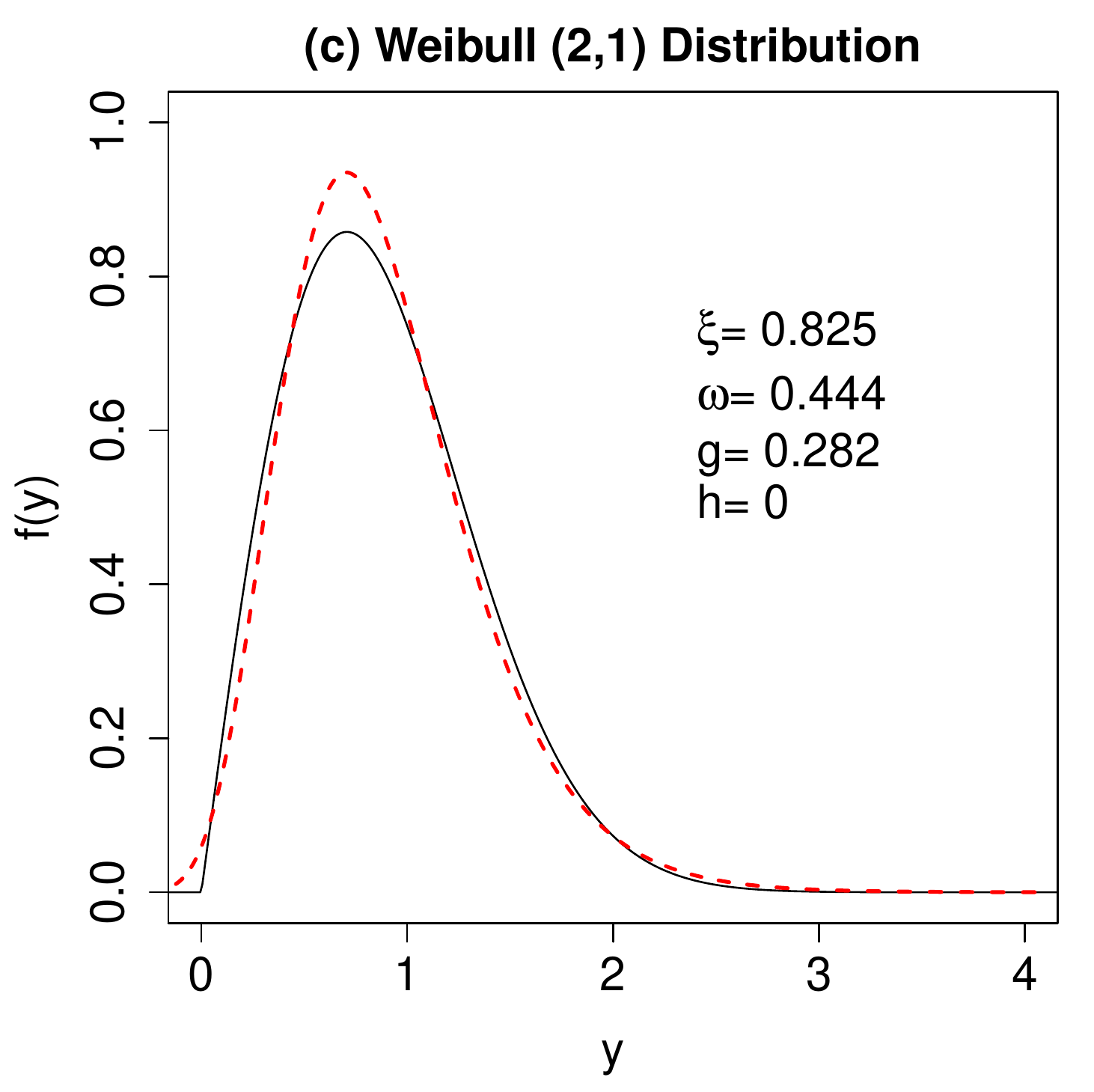}\\
\end{tabular}
\caption{Illustrations of using Tukey's $g$-and-$h$ density function (red dashed line) to approximate other density functions (black solid line): (a) Standard Cauchy distribution; (b) Student $t$-distribution with degrees of freedom 2; (c) Weibull distribution with shape parameter 2 and scale parameter 1.}
\label{fig-5}
\end{figure}

The flexibility of Tukey's $g$-and-$h$ family makes it a powerful tool to model real data arisen from many research areas. Examples of applications include modelling short interest rate distributions (Dutta \& Babbel, 2002), air pollution data (Rayner \& MacGillivray, 2002a, 2002b), extreme wind speed (Field, 2004), the value-at-risk of stock prices (Jim\'{e}nez \& Arunachalam, 2011), operational risk (Degen et al., 2007), and so on. Field \& Genton (2006) proposed a multivariate version of Tukey's $g$-and-$h$ distribution and used it to study data on Australian athletes and on wind speed. He~\& Raghunathan (2006, 2012) proposed to use Tukey's $g$-and-$h$ distribution to perform multiple imputations for missing data. Despite its popularity, there remain two unsolved challenges associated with Tukey's $g$-and-$h$ distribution: the lack of optimal parameter estimation procedures and the lack of valid statistical inference tools. Since a small variation in $g$ and $h$ may result in significant changes of the shape of the distribution, an estimation procedure with high accuracy is crucial for applying Tukey's $g$-and-$h$ distribution to real data. Unfortunately, the most accurate maximum likelihood estimator is not available for Tukey's $g$-and-$h$ distribution because the inverse function of $\tau_{g,h}(\cdot)$ does not have a closed form unless $h=0$, which makes the likelihood function intractable. Rayner \& MacGillivray (2002a) proposed a method to numerically evaluate the log-density function of Tukey's $g$-and-$h$ distribution, but their approach lacks resistance (Hoaglin, 2010) in that it can be computationally expensive when the sample size is large and it may also be numerically unstable for a reason that we will discuss later. To bypass this computational challenge, the existing literature has mainly relied on estimation procedures involving matching sample quantiles (Hoaglin, 1985; Dutta \& Babbel, 2002) or sample moments (Majumder \& Ali, 2008) with their population counterparts. Recently, Xu et al. (2014) proposed a new estimation procedure named quantile least square approach to estimate the parameters. Although all these methods provide satisfactory parameter estimators in many applications, our numerical experience shows that they can be significantly less accurate than the maximum likelihood estimator. An additional problem with these methods compared to the maximum likelihood estimator is that their estimation accuracies depend on a pre-selected set of quantiles or moments, which can be subjective in practice. To the best of our knowledge, there has been no study on how to choose an ``optimal" set of quantiles/moments to sharpen the estimation accuracies for these methods.

A second long-standing challenge with Tukey's $g$-and-$h$ family is how to provide valid statistical inference for the parameters. Although Tukey's $g$-and-$h$ distributions were first introduced as a tool to explore the data, they also have the potential to be used as an inference tool for the underlying distribution. For example, one can make statistical inference on whether the underlying distribution is symmetric by testing the hypothesis $g=0$. While there have been numerous attempts to improve the estimation accuracy, Xu et al. (2014) were the first to derive the asymptotic distribution of their estimator. However, like many other quantile-based estimators, this asymptotic distribution also depends on a subjectively selected set of quantiles and it remains unclear how will this choice affect the validity of the subsequent inference, especially when the sample size is small. Furthermore, the limiting distribution of their estimator is only valid when the true value of $h$, say $h_0$, satisfies the condition $h_0>0$. In the special case of $h_0=0$, the inference becomes irregular and the limiting distribution can be much more complicated than a normal distribution. The reason is that the restriction $h\geq 0$, which is necessary to ensure the monotonicity of $\tau_{g,h}(\cdot)$, makes $h_0=0$ fall on the boundary of the parameter space. Therefore, when $h_0=0$, the regularity conditions in  Xu et al. (2014) will be violated and therefore the result of Xu et al. (2014) cannot be used to test hypotheses such as $h=0$. However, because of the special interpretations of the shape parameters $g$ and $h$, testing $g=0$ or $h=0$ may be of particular interest in many applications.

In this paper, we aim at removing the bottle-neck of sub-optimal estimation procedures and the lack of statistical inference tools for Tukey's $g$-and-$h$ distribution. By approximating the likelihood function using a much simpler tractable function, we are able to obtain a maximum approximated likelihood estimator for parameters of Tukey's $g$-and-$h$ distributions, which is shown to be as efficient as the true maximum likelihood estimator under mild conditions. In addition, we derive the limiting distribution of the proposed maximum approximated likelihood estimator and develop valid approximated likelihood ratio tests for a series of hypotheses for the shape parameters $g$ and $h$, regardless of the true value $h_0$ equals $0$ or not. Our simulation studies demonstrate that the proposed approach is much more efficient than the quantile-based estimators and reaches the same efficiency as that of the maximum likelihood estimator.

The rest of the paper is organized as follows. In Section~2, an efficient estimation approach based on an approximated likelihood function for Tukey's $g$-and-$h$ distribution is proposed and related computational issues are discussed. The asymptotic and finite sample properties of the proposed maximum approximated likelihood estimators are investigated in Section~3. In Section~4, simulation studies are conducted to evaluate the performance of the proposed estimation procedure and approximated likelihood ratio tests. An application of our methodology to air pollution data is presented
in Section~5. The article ends with a conclusion in Section~6 and all theoretical
results are collected in the Appendix.

\section{Parameter Estimation}

\subsection{Existing approaches}

Denote the parameter vector of Tukey's $g$-and-$h$ distribution by $\btheta=(\xi,\omega,g,h)^T$. The log-density function of the random variable $Y$ from transformation~(\ref{TGHdef}) can be written as
\be
\label{density}
\log f_{Y|\btheta}(y)=\log\phi\left\{\tau_{g,h}^{-1}\left(\frac{y-\xi}{\omega}\right)\right\}-\log\omega-\log\tau_{g,h}'\left\{\tau_{g,h}^{-1}\left(\frac{y-\xi}{\omega}\right)\right\},
\ee
where $\phi(\cdot)$ is the standard normal density function, and $\tau_{g,h}^{-1}(\cdot)$  and $\tau_{g,h}'(\cdot)$ are the inverse function and the first derivative function of $\tau_{g,h}(\cdot)$, respectively. Suppose that we have a random sample $\{y_1,\dots,y_n\}$ and let $\bY=(y_1,\dots,y_n)^T$. Then the maximum likelihood estimator $\hat{\btheta}_{mle,n}$ is obtained by maximizing the log-likelihood function
\be
\label{mle}
L_n(\btheta)=\sum_{i=1}^n\log f_{Y|\btheta}(y_i).
\ee
It is well known that under mild regularity conditions, the limiting distribution of $\hat{\btheta}_{mle,n}$ has the smallest variance. One can further utilize tools such as the likelihood ratio test to make statistical inference on $\btheta$. Unfortunately, since $\tau_{g,h}^{-1}(\cdot)$ does not have a closed form, numerically evaluating $L_n(\btheta)$ can be computationally expensive, especially when the sample size is large. For this reason, the existing literature has largely been focusing on searching for alternative estimators, two of such examples are given below.

For a pre-selected sequence, $0<p_1<p_2<\cdots<p_K<1$, denote by $\hat{q}_{p_1},\dots,\hat{q}_{p_K}$ the corresponding sample quantiles of $\{y_1,\dots,y_n\}$ and let $z_{p_k}=\Phi^{-1}(p_k)$ for $k=1,\dots,K$, where $\Phi^{-1}(\cdot)$ is the inverse of the $N(0,1)$ cumulative distribution function. The first approach aims at directly matching a sequence of sample quantiles and theoretical quantiles, which we refer to as the letter-value-based approach (Dutta \& Babbel, 2002). The letter-value-based estimator $\hat\btheta_{lv,n}=(\hat \xi_{lv}, \hat \omega_{lv}, \hat g_{lv},\hat h_{lv})^T$ is defined as: $\hat \xi_{lv}=\hat q_{1/2}$, $\hat g_{lv}=\mbox{median}_{k=1,\ldots,K}\{\hat g_k\}$, where $\hat g_k=-\frac{1}{z_{p_k}} \log \left( \frac{\hat{q}_{1-p_k}-\hat{q}_{1/2}}{\hat q_{1/2}-\hat q_{p_k}}\right)$ and finally $(\hat\omega_{lv}, \hat h_{lv})$ are obtained from the linear regression
\[
\log\left\{ \frac{\hat g_{lv}(\hat q_{p_k}-\hat q_{1-p_k})}{\exp(\hat g_{lv}z_{p_k})-\exp(-\hat g_{lv}z_{p_k})}\right\}=\log \omega + h z_{p_k}^2/2,\quad k=1,\dots,K.
\]

A second approach was proposed by Xu et al. (2014), named the quantile least square estimator $\hat\btheta_{qls,n}$, which is defined as the minimizer of the quantile least square loss function
\[
L_{qls}(\btheta)=\sum_{k=1}^K\left\{\hat{q}_{p_k}-q_{p_k}(\btheta)\right\}^2,
\]
where $q_{p_k}(\btheta)=\xi+\omega\tau_{g,h}(z_{p_k})$ is the theoretical $p_k$th quantile of the random variable $Y$.

As one can see, both $\hat\btheta_{lv,n}$ and $\hat\btheta_{qls,n}$ rely on a suitable choice of $p_{k}$'s. While this choice is largely subjective in practice, there has not yet been research on how this choice may impact the efficiency and the limiting distribution of resulting estimators.

\subsection{A second representation of $\log f_{Y|\btheta}(y)$}

The main difficulty for obtaining the maximum likelihood estimator lies in the lack of closed form for $\log f_{Y|\btheta}(y)$. A natural idea is to find an explicitly computable function that can approximate $\log f_{Y|\btheta}(y)$ well. To do so, we first define $p_{\btheta}(y)=F_{Y|\btheta}(y)$ and $z_{p_{\btheta}(y)}=\Phi^{-1}\{p_{\btheta}(y)\}$, where $F_{Y|\btheta}(\cdot)$ is the cumulative distribution function of $Y$ with a parameter vector $\btheta$. By this definition, it is straightforward to show that $z_{p_{\btheta}(y)}=\tau_{g,h}^{-1}(\frac{y-\xi}{\omega})$. Then $\log f_{Y|\btheta}(y)$ in (\ref{density}) can be written as a function of $z_{p_{\btheta}(y)}$; that is,
\be
\label{density1}
\begin{split}
\varphi_{\btheta}\{z_{p_{\btheta}(y)}\}&= \log f_{Y|\btheta}(y) =\log\phi\{z_{p_{\btheta}(y)}\}-\log\omega-\log\tau_{g,h}'\{z_{p_{\btheta}(y)}\}\\&=-\frac{1+h}{2}z_{p_{\btheta}(y)}^2-\log\left[\exp\{gz_{p_{\btheta}(y)}\}+g^{-1}\{\exp(gz_{p_{\btheta}(y)})-1\}hz_{p_{\btheta}(y)}\right]-\log\omega-\frac{1}{2}\log(2\pi).
\end{split}
\ee
\begin{figure}[b!]
\centering
\begin{tabular}{cc}
\includegraphics[width=0.40\textwidth]{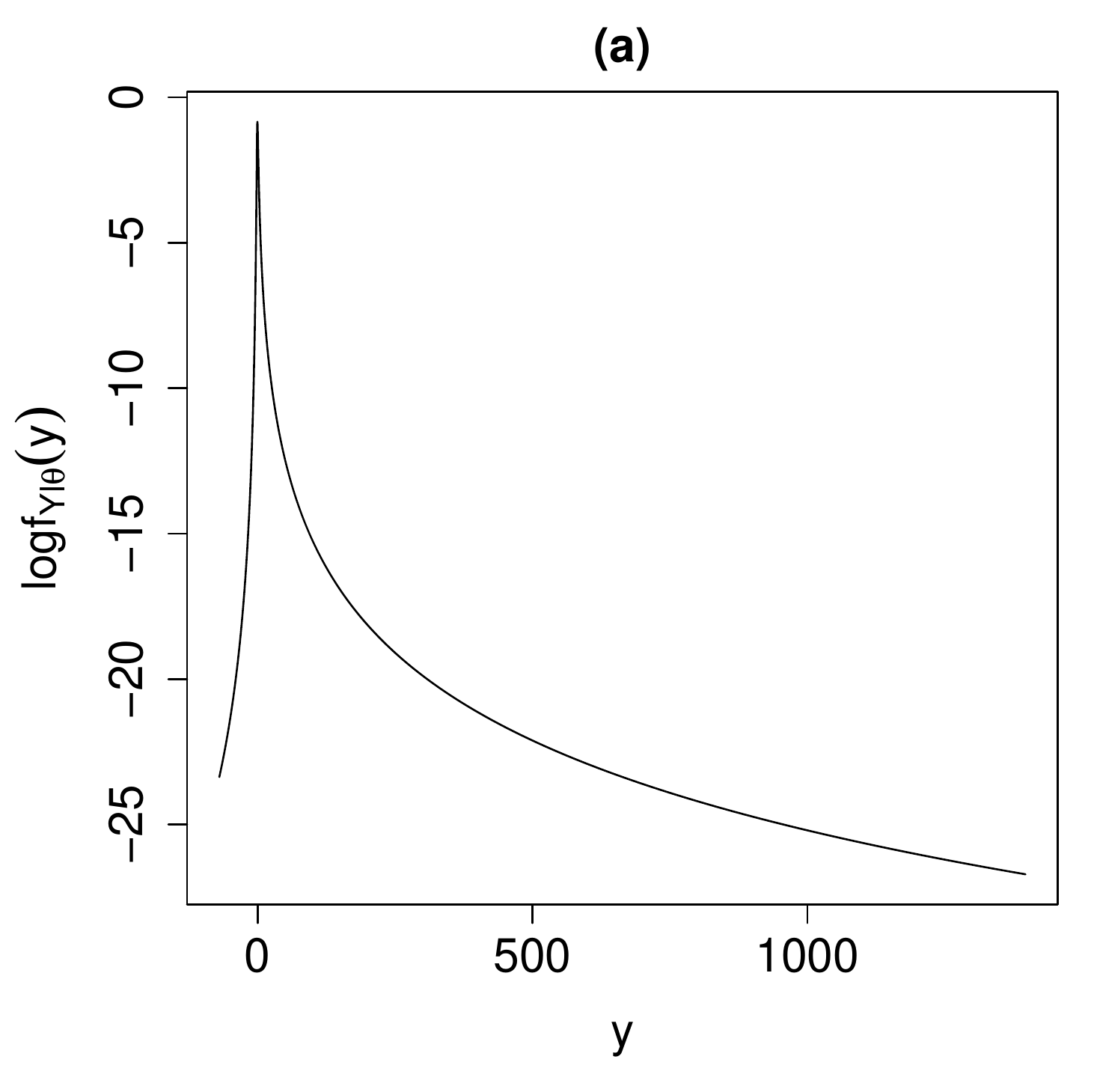} &
\includegraphics[width=0.40\textwidth]{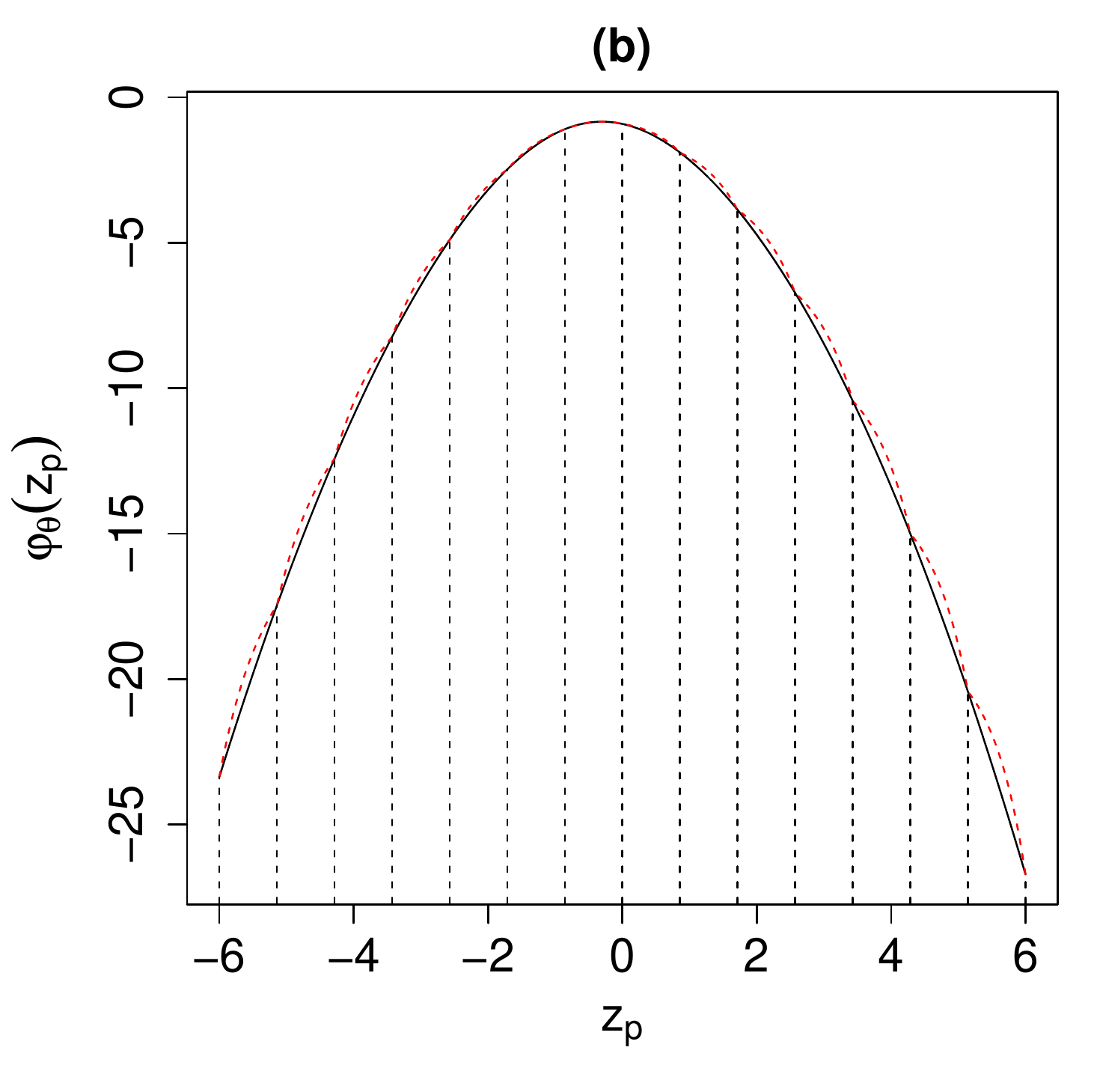}\\
\includegraphics[width=0.40\textwidth]{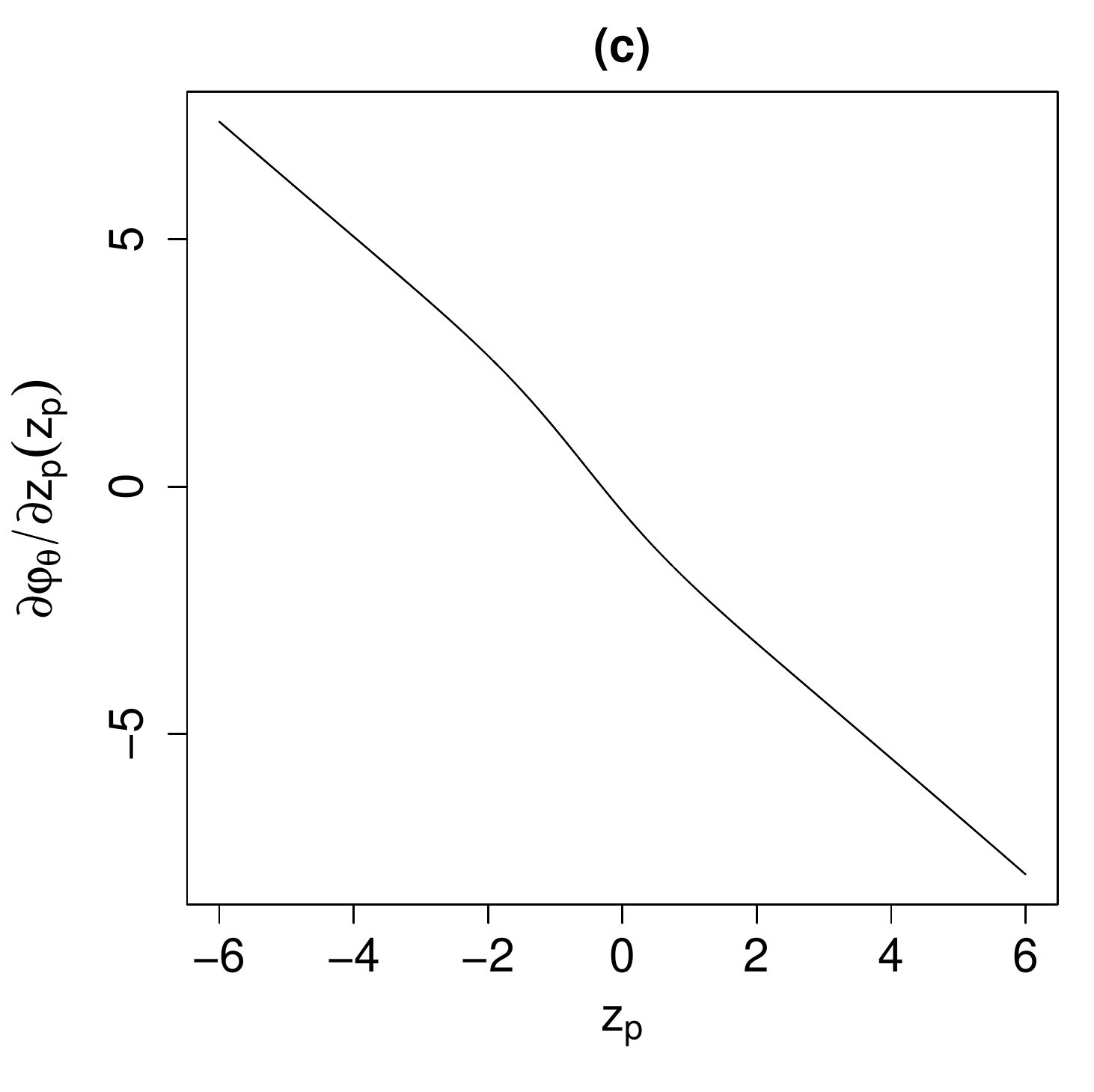} &
\includegraphics[width=0.40\textwidth]{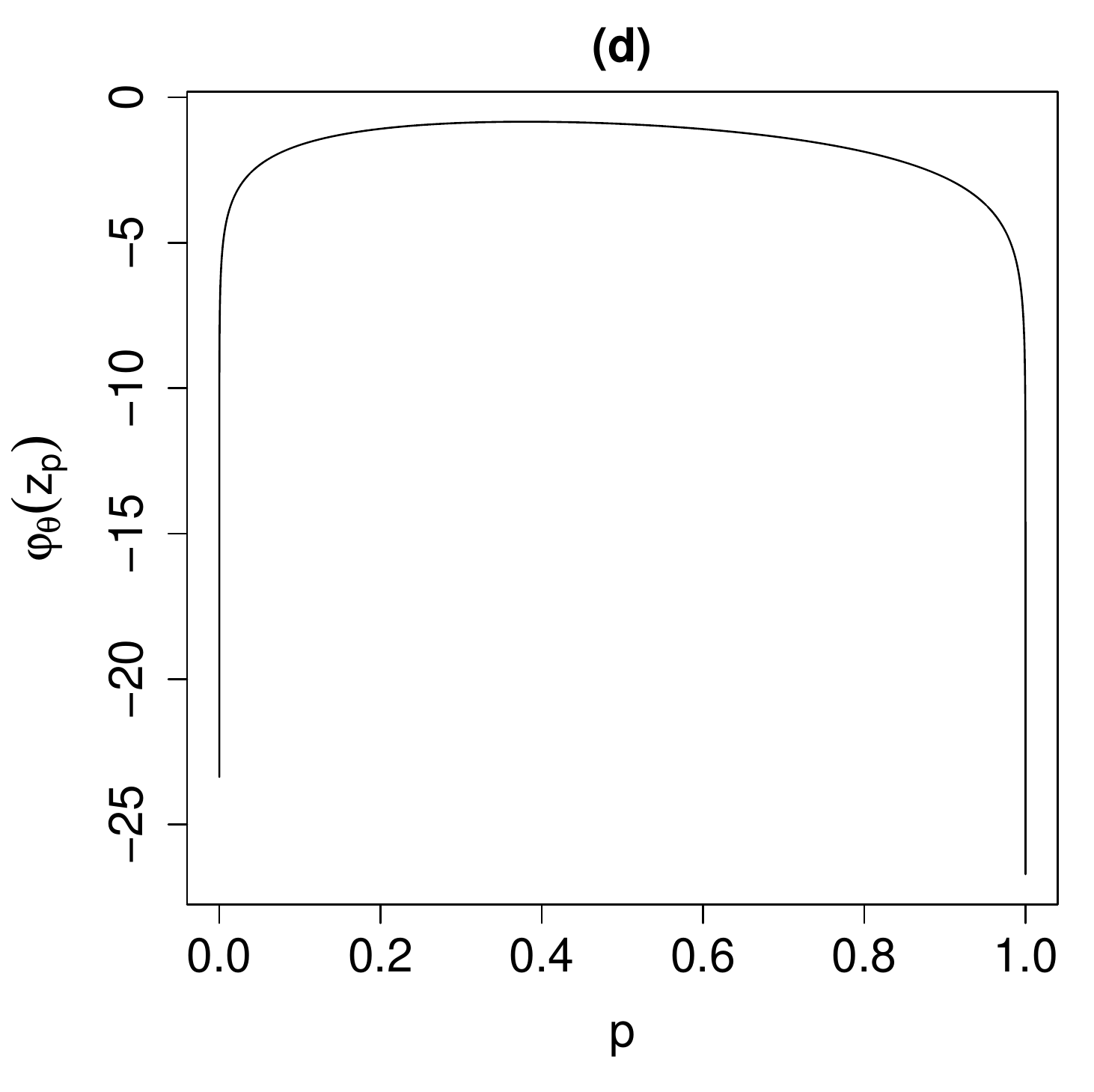}\\
\end{tabular}
\caption{The log-density function of a Tukey's $g$-and-$h$ distribution with $(\xi,\omega,g,h)=(0,1,0.5,0.2)$. (a) $\log f_{Y|\btheta}(y)$ as a function of $y$; (b) $\varphi_{\btheta}(z_{p})$ (solid line) vs $\tilde{\varphi}_{\btheta}(z_{p})$ (dashed line, $K_n=15$,  $b_n=6$) as functions of $z_p$; (c) $\frac{\partial \varphi_{\btheta}}{\partial^{\dag} z_p}(z_p)$ as a function of $z_p$; (d) $\varphi_{\btheta}(z_{p})$ as a function of $p$.}
\label{fig-1}
\end{figure}
The notation $z_{p_{\btheta}(y)}$ is used here to emphasize that $z_{p_{\btheta}(y)}$ is an unknown quantity depending on $\btheta$ and $y$. For simplicity, from now on, we write $z_{p_{\btheta}(y)}$ as $z_p$ whenever there is no ambiguity. Recall that $p$ is the cumulative distribution function of $Y$ evaluated at the observed value $y$. Figures~\ref{fig-1}(a) and (b) depict an example of $\log f_{Y|\btheta}(y)$, or equivalently $\varphi_{\btheta}(z_{p})$, as a function $y$ and $z_p$, respectively. From Figures~\ref{fig-1}(a)-(b), we can see that the observed data $\{y_1,\dots,y_n\}$ are sparsely distributed in $(-\infty,\infty)$. As a function of $y$, $\log f_{Y|\btheta}(y)$ varies rapidly on the left-hand side of the sharp spike in the middle. On the contrary, within a bounded interval, $\varphi_{\btheta}(z_{p})$ is a slowly varying function of $z_p$, which makes it easier to approximate using some numerical method. Furthermore, although $\varphi_{\btheta}(z_{p})$ is also defined on $(-\infty,\infty)$ in theory, with a finite sample size $n$, one only has to focus on a bounded interval $[-b_n,b_n]$ for some $b_n>0$. To see this, consider the case when $\btheta=\btheta_0$ with $\btheta_0=(\xi_0,\omega_0,g_0,h_0)^T$ being the true value of $\btheta$ that generated the data and choose $b_n=\Phi^{-1}\{1-1/(n\log n)\}$. In this case, the $p_i=F_{Y|\btheta_0}(y_i)$'s follow a uniform distribution on $[0,1]$ and the corresponding $z_{p_i}$'s follow the $N(0,1)$ distribution. Straightforward algebra yields that
\[
{\rm P}\left(\max_{1\leq i\leq n}z_{p_i}>b_n\right)=1-\Phi(b_n)^n=1-\left(1-\frac{1}{n\log n}\right)^n\to 0\text{ as } n\to\infty,
\]
where $\Phi(\cdot)$ is the standard normal cumulative distribution function. In this sense, with a sample size $n=10,000$, using $b_n=4.25$ is already a good choice. Therefore, for $\btheta$ in a sufficiently small neighborhood of $\btheta_0$ and a finite sample size $n$, it is reasonable to treat $\varphi_{\btheta}(z_{p})$ as a function defined on a bounded support $[-b_n,b_n]$. Our numerical experience shows that it is generally sufficient to take $b_n=10$ in practice.

\subsection{Maximum approximated likelihood estimator}

In this section we develop a numerical
algorithm for computing the log-density value, $\log f_{Y|\btheta}(y)=\varphi_{\btheta}(z_{p})$, at an observation $y$ for a known parameter
vector $\btheta$, which cannot be computed directly because the value of $z_p$ is unknown due to the fact that $\tau_{g,h}^{-1}(\cdot)$ does not have a closed form. To overcome this, we propose the following approach to approximate the value of $z_p$. With a sample size $n$ and a pre-given $b_n>0$, we first introduce $K_n$ equally spaced knots over the interval $[-b_n,b_n]$, denoted as $-b_n=\textsf{Z}_1<\textsf{Z}_2<\dots<\textsf{Z}_{K_n}=b_n$, and then compute the corresponding knots in the transformed scale as $\textsf{Y}_{k,\btheta}=\xi+\omega\tau_{g,h}(\textsf{Z}_k)$, $k=1,\dots,K_n$. For a given $y\in [\textsf{Y}_{1,\btheta},\textsf{Y}_{K_n,\btheta}]$, we find the knot $\textsf{Z}_k$ such that $\textsf{Y}_{k,\btheta}\leq y<\textsf{Y}_{k+1,\btheta}$. The monotonicity of $\tau_{g,h}(\cdot)$ ensures that the $z_p$ associated with $y$ must satisfy $\textsf{Z}_k\leq z_p <\textsf{Z}_{k+1}$. Instead of computing $z_p$ by numerically solving the equation $y=\xi+\omega\tau_{g,h}(z_p)$, we use the following linear approximation
\be
\label{zhat}
\tilde{z}_{p,k}=\textsf{Z}_k+\frac{y-\textsf{Y}_{k,\btheta}}{\textsf{Y}_{k+1,\btheta}-\textsf{Y}_{k,\btheta}}(\textsf{Z}_{k+1}-\textsf{Z}_k)\quad \text{if}\quad \textsf{Y}_{k,\btheta}\leq y<\textsf{Y}_{k+1,\btheta}.
\ee
Because $|\tilde{z}_{p,k}-z_p|\leq 2b_n/K_n$ by design, we can expect that if $K_n$ is sufficiently large such that $b_n/K_n\to 0$, $\tilde{z}_{p,k}$ should approximate $z_p$ well. Then, for any observed value $y$ such that $y=\xi+\omega\tau_{g,h}(z_p)$, we can define an approximation function for $\varphi_{\btheta}(z_p)$ as
\be
\label{density2}
\tilde{\psi}_{\btheta}(y)=\tilde{\varphi}_{\btheta}(z_p)=\sum_{k=1}^{K_n-1}\varphi_{\btheta}(\tilde{z}_{p,k})I_{[\textsf{Y}_{k,\btheta},\textsf{Y}_{k+1,\btheta}]}(y),
\ee
where $I_A(y)=1$ if $y\in A$ and $0$ otherwise. In Figure~{\ref{fig-1}}(b), we can see that $\tilde{\varphi}_{\btheta}(z_p)$ is a piecewise convex function yielding a good approximation, even though only $K_n=15$ knots were used in this example. To see why this is the case, it is straightforward to show that, for any given $\btheta$,
\[
\sup_{y\in [\textsf{Y}_{1,\btheta},\textsf{Y}_{K_n,\btheta}]}\left|\varphi_{\btheta}(z_p)-\tilde{\varphi}_{\btheta}(z_p)\right|\leq \frac{2b_n}{K_n}\sup_{z_p\in [-b_n,b_n]}\left|\frac{\partial \varphi_{\btheta}}{\partial^{\dag} z_p}(z_p)\right|,
\]
where the notation $\partial^{\dag}$ is used here to distinguish it from the usual partial derivative. More specifically, by using $\partial^{\dag}$, we treat $z_p$ as an argument of the function $\varphi_{\btheta}(z_p)$ that is independent of $\btheta$, even though $z_p$ is actually a function of $\btheta$. For example,
 \[
 \label{diff1}
\frac{\partial \varphi_{\btheta}}{\partial^{\dag} z_p}(z_p)=-(1+h)z_p-g-\frac{g^{-1}\{\exp(gz_p)-1\}+z_p}{\exp(gz_p)+g^{-1}\{\exp(gz_p)-1\}hz_p}h.
 \]
  Figure~{\ref{fig-1}}(c) illustrates that for the same example used in Figure~{\ref{fig-1}}(b), $\frac{\partial \varphi_{\btheta}}{\partial^{\dag} z_p}(z_p)$ is essentially bounded within the interval $[-b_n,b_n]$ and therefore $\sup_{y\in [\textsf{Y}_{1,\btheta},\textsf{Y}_{K_n,\btheta}]}\left|\varphi_{\btheta}(z_p)-\tilde{\varphi}_{\btheta}(z_p)\right|\to 0$ as long as $b_n/K_n\to 0$ when $n\to\infty$.

Given a random sample $\{y_1,\dots,y_n\}$, we can repeat the above process to calculate $p_i=F_{Y|\btheta}(y_i)$ and $z_{p_i}=\Phi^{-1}(p_i)$ for each data point. Then, the log-likelihood function~(\ref{mle}) can be approximated by the function
\be
\label{male}
\tilde{L}_n(\btheta)=\begin{cases}\sum_{i=1}^{n} \tilde{\psi}_{\btheta}(y_i),& \text{if}\quad \textsf{Y}_{1,\btheta}\leq y_{\min}<y_{\max}\leq \textsf{Y}_{K_n,\btheta}, \\-\infty,&\text{otherwise,}
\end{cases}
\ee
where $\tilde{\psi}_{\btheta}(y)$ is defined in~(\ref{density2}) and $y_{\min},y_{\max}$ are the smallest and largest observed values. The maximum approximated likelihood estimator, $\hat{\btheta}_{male,n}$, is defined as the maximizer of (\ref{male}).

\subsection{Some computational issues}

 Rayner \& MacGillivray (2002a) proposed a similar ``change of support" approach to numerically approximate $\log f_{Y|\btheta}(y_i)$ for each $y_i$, $i=1,\dots,n$. The idea is essentially to numerically solve $y=\xi+\omega\tau_{g,h}(z_p)$ for $p$ and then plug the solution $\hat{p}$ back into $\log f_{Y|\btheta}(y)=\varphi_{\btheta}(z_p)$. However, this approach can be computationally expensive when $n$ is large. In addition, as illustrated in Figure~{\ref{fig-1}}(d), $\varphi_{\btheta}(z_p)$ as a function of $p$ is rather steep on both ends of the interval $[0,1]$, which is very different from Figure~{\ref{fig-1}}(b). Consequently, for $p$ close to $0$ or $1$, a small error in $\hat{p}$ can lead to non-negligible errors when evaluating $\varphi_{\btheta}(z_p)$ using $\varphi_{\btheta}(z_{\hat{p}})$. In this sense, Rayner \& MacGillivray (2002a)'s method can also be numerically unstable.

The computational complexity of $\tilde{L}_n(\btheta)$ is equivalent to assigning $n$ data points to $K_n-1$ bins in a histogram, which can be efficiently implemented using some bucket sort algorithm (Corwin \& Logar, 2004). The average computational complexity for such an algorithm can be easily achieved as $O(n+K_n)$, which is fast even for very large $n$ and $K_n$. In this paper, we use the function {\tt .bincode} from the R software (R Development Core Team, 2014) to implement this algorithm. Based on Theorems 1-2 below, to guarantee estimation efficiency, we need to ensure that $K_n$ is large enough so that Condition 5 in the Appendix is met. Our numerical experience indicates that using $K_n=\max(1000,n)$ is sufficient for most applications. With such a choice of $K_n$, our simulation studies in Section 4 demonstrate that the proposed estimation method can be several hundred times faster than the numerical maximum likelihood estimation approach proposed in Rayner \& MacGillivray (2002a).

Another important issue in maximizing $\tilde{L}(\btheta)$ defined in~(\ref{male}) is the choice of initial values for $\btheta$. In the Supplementary Material, we give the Gradient and the Hessian matrix of $\tilde{L}(\btheta)$, which can be utilized for any standard optimization routine. We use the ``L-BFGS-B" method (Byrd et al. 1995) provided in the R function {\tt optim} to maximize $\tilde{L}(\btheta)$ since it is a constrained optimization problem due to the restriction $h\geq 0$. To find $\hat\btheta_{male,n}$, the initial value $\btheta^0$ must be chosen such that Condition~1 in the Appendix is met, that is, $\textsf{Y}_{1,\btheta^0}\leq y_{\min}<y_{\max}\leq \textsf{Y}_{K_n,\btheta^0}$. Based on our discussion in Subsection~2.1, we need essentially to find a $\btheta^0$ that is reasonably close to the true value $\btheta_0$. Fortunately, such a requirement can be easily fulfilled by taking $\btheta^0$ as the letter-value-based estimator $\hat\btheta_{lv,n}$ or the quantile least square estimator $\hat\btheta_{qls,n}$. In this paper, we choose $\hat\btheta_{lv,n}$ as the initial value of our algorithm.

Finally, we would like to point out that, unlike the quantile-based estimators $\hat\btheta_{lv,n}$ and  $\hat\btheta_{qls,n}$, whose performances depend on a subjective choice of quantiles, the estimation accuracy of the maximum approximated likelihood estimator $\hat\btheta_{male,n}$ does not depend on the positions of the chosen knots as long as the number of knots $K_n$ is sufficiently large. This is reflected in the fact that the results of Theorems 1-2 do not depend on $K_n$ at all. In addition, we have proposed to use equally spaced knots merely for the simplicity of our technical investigations. Generally speaking, any knot sequence $-b_n=\textsf{Z}_1<\textsf{Z}_2<\dots<\textsf{Z}_{K_n}=b_n$ satisfying the condition that $\inf_{1\leq i\leq K_n-1}|\textsf{Z}_{i+1}-\textsf{Z}_{i}|\to 0$ sufficiently fast as $n\to\infty$ will do the job.

\section{Statistical Inference}

It is often desirable to provide an uncertainty measure for a point estimator in statistical research, which remains a challenge when fitting Tukey's $g$-and-$h$ distribution to data. For the popular letter-value-based approach, the asymptotic distribution of $\hat{\btheta}_{lv,n}$ remains unknown. Although the quantile least square approach proposed by Xu et al. (2014) results in an asymptotic normal distribution for $\hat{\btheta}_{qls,n}$, there remain two problems to be addressed. First, the asymptotic distribution of $\hat{\btheta}_{qls,n}$ depends on a preselected set of $p_k$'s and it remains unknown how this choice affects the subsequent statistical inference. Furthermore, their theory relies on the implied assumption that $h>0$ and does not hold if the true value of $h$ is $0$. However, for Tukey's $g$-and-$h$ distribution, $h=0$ is a special case of particular interest. On the one hand, by testing $h=0$, one can tell whether or not the data have heavy tails. On the other hand, if we have significant evidence to believe that the true value of $h$ is $0$, then letting $h=0$ would make the function $\tau_{g,h}(\cdot)$ invertible and thus makes the maximum likelihood estimator achievable.

\subsection{Asymptotic properties of $\hat\btheta_{male,n}$}
 Denote by $\Omega=\Bbb{R}\times (0,\infty)\times \Bbb{R}\times [0,\infty)$ the parameter space of $\btheta=(\xi,\omega,g,h)^T$ and define $U_n(\btheta)=\partial L_n(\btheta)/\partial \btheta$, $\tilde{U}_n(\btheta)=\partial \tilde{L}_n(\btheta)/\partial \btheta$ and $I_n(\btheta)=-\partial^2L_n(\btheta)/\partial \btheta\partial\btheta^T$. Let $I(\btheta_0)$ be the expected value of $I_n(\btheta)$ when $\btheta=\btheta_0$ and define a random vector $\bZ=(Z_1,Z_2,Z_3,Z_4)^T\sim N_4(\bm 0,I^{-1}(\btheta_0))$, where $N_4(\bm 0, I^{-1}(\btheta_0))$ is a 4-dimensional multivariate normal distribution with mean $\bf 0$ and covariance matrix $I^{-1}(\btheta_0)$. The following theorem gives the asymptotic distribution of $\hat{\btheta}_{male,n}$ under the scenario $h_0=0$ or $h_0>0$. The proof is given in the Appendix.

\begin{theo}
\label{thm1}(Asymptotic Distribution) Under Conditions 1-5 in the Appendix, as $n\to\infty$ and $K_n\to\infty$, we have: (a) If the true value $h_0>0$, then $\sqrt{n}(\hat{\btheta}_{male,n}-\btheta_0)\rightarrow \bZ$ in distribution; (b) If the true value $h_0=0$, then
\[
\sqrt{n}(\hat{\btheta}_{male,n}-\btheta_0)\rightarrow\left(\begin{array}{c}Z_1\\Z_2\\Z_3\\Z_4\end{array}\right)I(Z_4>0)+\left(\begin{array}{c}Z_1-I^{14}/I^{44}Z_4\\Z_2-I^{24}/I^{44}Z_4\\Z_3-I^{34}/I^{44}Z_4\\0\end{array}\right)I(Z_4<0)\quad\text{in distribution},
\]
where $I^{ij}$ is the $ij$-th element of $I^{-1}(\btheta_0)$, $\btheta=(\xi,\omega,g,h)^T$ and $\btheta_0=(\xi_0,\omega_0,g_0,h_0)^T$.
\end{theo}

Theorem~{\ref{thm1}} essentially states that when $h_0>0$, if $K_n\to\infty$ fast enough as $n\to\infty$, $\hat{\btheta}_{male,n}$ has the same asymptotic distribution as $\hat{\btheta}_{mle,n}$ and reaches the Cram\'{e}r-Rao efficiency lower bound. This is confirmed by our simulation example 1 in Section~4. When $h_0=0$, asymptotically, $\hat{\btheta}_{male,n}$ has a mixture normal distribution and there are $50\%$ of times that $h$ will be estimated as exactly $0$. In the Supplementary Material, we give a detailed description of how to compute $U_n(\btheta)$ and $I_n(\btheta)$ numerically and thus approximate the information matrix $I(\btheta_0)$ by $I_n(\hat{\btheta}_{male,n})$.

Based on Theorem~1, we can see that it is critical to determine whether $h_0=0$ or not because the limiting distributions of $\hat{\btheta}_{male,n}$ under these two scenarios are radically different. So the next question becomes: can we perform a hypothesis test for $H_0$: $h=0$? Fortunately, the answer is yes. Let $\Omega_0\subseteq \Omega$ be a subset of the parameter space of $\btheta$. We define the approximated likelihood ratio test statistic for testing the null hypothesis $H_0:\btheta\in\Omega_0$ as
\be
\label{alrt}
D_n=-2\{\sup_{\btheta\in\Omega_0}\tilde{L}_{n}(\btheta)-\sup_{\btheta\in\Omega}\tilde{L}_{n}(\btheta)\}.
\ee
In this paper, we consider three types of $\Omega_0$'s: $\Omega_0^{(1)}=\Bbb{R}\times(0,\infty)\times \{0\}\times [0,\infty)$, $\Omega_0^{(2)}=\Bbb{R}\times(0,\infty)\times \Bbb{R}\times\{0\}$ and $\Omega_0^{(3)}=\Bbb{R}\times(0,\infty)\times \{0\}\times\{0\}$, which correspond to testing $H_0^{(1)}:g=0$, $H_0^{(2)}:h=0$, $H_0^{(3)}:h=g=0$, respectively. The following theorem gives the limiting distributions of $D_n$ under these three null hypotheses. The proof is given in the Appendix.
\begin{theo}\label{thm2}(Approximated Likelihood Ratio Tests, ALRT) Under Conditions 1-5 in the Appendix, as $n\to\infty$ and $K_n\to\infty$,
\begin{enumerate}
\item Under $H_0^{(1)}: g=0$, one has $D_n\to D\sim \chi^2_1$ in distribution;
\item Under $H_0^{(2)}: h=0$, one has $D_n\to D\sim 0.5\chi^2_0+0.5\chi^2_1$ in distribution;
\item Under $H_0^{(3)}: g=h=0$, one has $D_n\to D\sim 0.5\chi^2_1+0.5\chi^2_2$ in distribution;
\end{enumerate}
where $\chi^2_{df}$ is the chi-square distribution with $df$ degrees of freedom and $0.5\chi^2_{df_1}+0.5\chi^2_{df_2}$ stands for a mixture distribution of $\chi^2_{df_1}$ and $\chi^2_{df_2}$ with $50\%$ of each component. In particular, $\chi^2_0$ represents the distribution with a point mass at $0$, i.e. $P(X=0)=1$ if $X\sim\chi^2_0$.
\end{theo}

The critical values of the mixed $\chi^2$ distribution can be found using the function {\tt qchibarsq} in the R package ``emdbook" (Bolker, 2008). From Theorem~\ref{thm2}, it is interesting that $D_n$ for testing $H_0^{(1)}: g=0$ has the the same limiting distribution, whether or not the value of $h_0$ is $0$. Under the null hypotheses $H_0^{(2)}: h=0$ or $H_0^{(3)}: g=h=0$, the first component of the asymptotic distributions of $D_n$ comes from the cases when $\hat{h}_{male,n}=0$, which should happen, asymptotically£¬ $50\%$ of times by Theorem~{\ref{thm1}}.

\subsection{Finite sample performance of ALRT tests}

Although by Theorem~{\ref{thm1}}, when the sample size $n\to\infty$, the occurrence rate of $h$ being estimated as $0$, denoted by $P_{0,n}$, should approach $0.5$, our empirical studies show that $P_{0,n}$ is almost always bigger than $0.5$ when $n$ is small. In our simulation studies, we noticed that, for a fixed $n$, $P_{0,n}$ is closely related to another quantity, $C_{0,n}=P(S_n<0)$, where $S_n$ is the fourth entry in the vector $n^{-1/2}U_n(\btheta_0)$. As illustrated in Figure~\ref{fig-3} (b), when the sample size $n$ grows, the difference between $P_{0,n}$ and $C_{0,n}$ becomes smaller and smaller. An intuitive explanation is that if $S_n<0$, the approximated likelihood function $\tilde{L}(\btheta)$ cannot be improved by moving the parameter $h$ away from $0$ toward $\infty$, given that the other parameters are fixed at $(\xi_0,\omega_0,g_0)$, and therefore $h=0$ should be the resulting solution. Although the explicit relationship between $P_{0,n}$ and $C_{0,n}$ remains unclear, studying the behavior of $C_{0,n}$ will provide some insights on why $P_{0,n}\to 0.5$ at such a slow rate.

In the special case of $g_0=h_0=0$,  some tedious algebra yields that $S_n=-n^{-1/2}\sum_{i=1}^n(z_i^4/3-z_i^2)$ with $z_i$'s being independent samples from the $N(0,1)$ distribution. By the Central Limit Theorem, $S_n$ converges in distribution to a normal distribution with mean $0$ and thus $C_{0,n}=P(S_n<0)\to 0.5$ as $n\to\infty$. However, even for a large $n$, the finite sample distribution of $S_n$ can still be severely skewed to the left due to the existence of the term $z_i^4$ in the summation. As a result, even for a large $n$, we can still have $C_{0,n}>0.5$. Our simulation studies show that this left-skewness in $S_n$ becomes even more sever when $|g_0|$ deviates from $0$. Because of the association between $P_{0,n}$ and $C_{0,n}$, this partially explains why $P_{0,n}$ is always greater than $0.5$.

Next, we use a small simulation study to illustrate this phenomenon. The data were generated using model~(\ref{TGHdef}) with $\xi_0=\omega_0=3$, $h_0=0$ and $g_0=-0.5,-0.4,\dots,0.4,0.5$. In Figure~\ref{fig-3}, for each $n$, we plot the values of $P_{0,n}$ and $C_{0,n}$ obtained using different values of $g_0$ in a boxplot. We can see that $P_{0,n}$ is always greater than $0.5$ regardless of the values of $g_0$ and $n$ and that it slowly regresses to $0.5$ as $n$ increases. In addition, as $n$ increases, the difference $P_{0,n}-C_{0,n}$ consistently decreases for various values of $g_0$, indicating strong association between $P_{0,n}$ and $C_{0,n}$.
\begin{figure}[ht!]
\centering
\begin{tabular}{cc}
\includegraphics[width=0.40\textwidth]{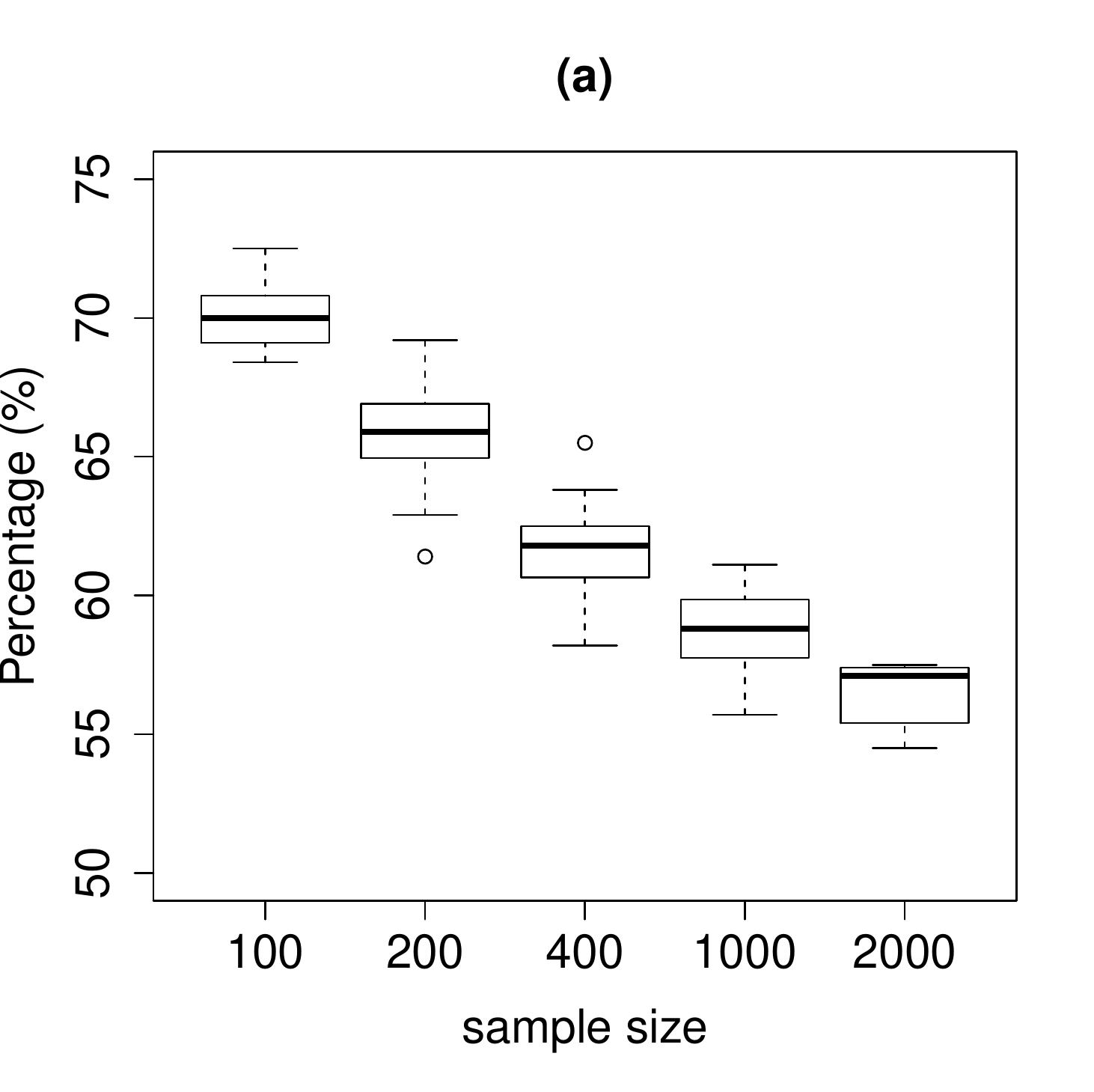} &
\includegraphics[width=0.40\textwidth]{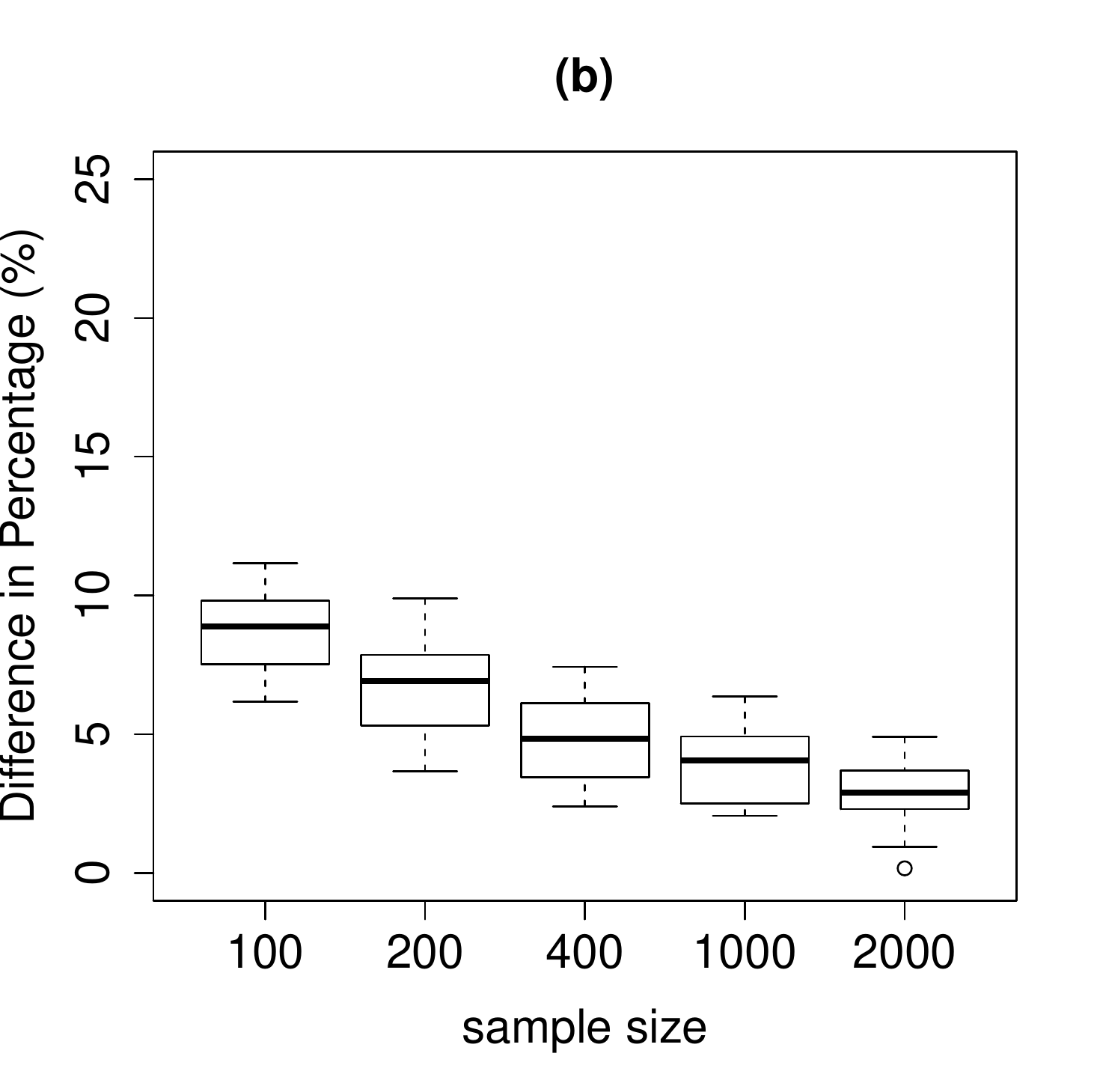}
\end{tabular}
\caption{(a) Percentage of times, $P_{0,n}$, when $\hat{h}_{male,n}=0$ with $h_0=0$; (b) Difference between $P_{0,n}$ and $C_{0,n}$: $P_{0,n}-C_{0,n}$ vs sample size $n$.}
\label{fig-3}
\end{figure}

Such a phenomenon indicates that, for a finite $n$, the critical values based on the asymptotic distributions in Theorem~\ref{thm2} are often too large for the null hypotheses $H_0^{(2)}$ and $H_0^{(3)}$. For example, to test $H_0^{(2)}$, if the estimated value for $h$ is $\hat h_{male,n}=0$ under the alternative hypothesis, the resulting approximate likelihood ratio test statistic $D_n=0$. And if $\hat h_{male,n}=0$ occurs more than $50\%$ of times, then the finite sample distribution of $D_n$ consists of more than $50\%$ of $0$'s when $H_0^{(2)}$ is true, which means we will fail to reject $H_0^{(2)}$ more often than $(1-\alpha)100\%$ of times using the $(1-\alpha)$th quantile of the $0.5\chi^2_0+0.5\chi^2_1$ distribution. Similar arguments apply to $H_0^{(3)}$ too. This explains why the size of the approximated likelihood ratio tests for $H_0^{(2)}$ and $H_0^{(3)}$ are often smaller than the nominal level in the simulation example 2. However, even with such conservativeness, the approximated likelihood ratio tests for $H_0^{(2)}$ and $H_0^{(3)}$ still display significant local powers.

\section{Simulation Studies}

In this section, we use some numerical examples to show the effectiveness of our method. The data were generated using model~(\ref{TGHdef}) with different values of $\btheta_0=(\xi_0,\omega_0,g_0,h_0)^T$ and five sample sizes $n=100,200,400,1000,2000$. All summary statistics were based on $1000$ repetitions and for each repetition, we used $b_n=10$ and $K_n=\max(1000,n)$.

\subsection{Evaluation of the estimation accuracy}

In the simulation example 1, we mainly focused on comparing estimation accuracies for different methods described in Section 2. Three alternative approaches were considered: the letter-value-based method (LV), the quantile least square method (QLS) and the numerical maximum likelihood estimation (NMLE, Rayner \& MacGillivray, 2002a). For the LV method, six quantiles were used: $\bm p_{lv}=(0.005,0.01,0.025,0.05,0.10,0.25)^T$. For the QLS method, 10 quantiles were used with $p_i = \frac{i-1/3}{m+1/3}$, $i=1,\dots, m$ and $m=10$. For the NMLE method, the function {\tt uniroot} form the R package was used to numerically solve $y=\xi+\omega\tau_{g,h}(z_p)$ for $p$.

We set $\btheta_0=(\xi_0,\omega_0,g_0,h_0)^T=(3,3,0.5,0.2)^T$ and the empirical standard errors based on $1000$ simulation runs are summarized in Figure~{\ref{fig-2}}. The asymptotic standard error derived from Theorem~{\ref{thm1}}, which is also the Cram\'{e}r-Rao efficiency lower bound of the maximum likelihood estimator (MLE), is reported as well. We can see that the maximum approximated likelihood estimator (MALE) is uniformly more accurate than both LV and QLS methods for all sample sizes. In terms of estimation efficiencies, the NMLE and MALE methods are almost identical and reach the efficiency lower bounds, which confirms the theoretical findings in Theorem~{\ref{thm1}}. In addition, the LV method did a better job at estimating the parameters $g$ and $h$ than the QLS method did, while it seems to be the reverse for estimating $\xi$ and~$\omega$. Additional simulation results using different values of $\btheta_0$ can be found in the Supplementary Material, where the findings are all consistent with Figure~\ref{fig-2}.
\begin{figure}[h!]
\centering
\begin{tabular}{cc}
\includegraphics[width=0.40\textwidth]{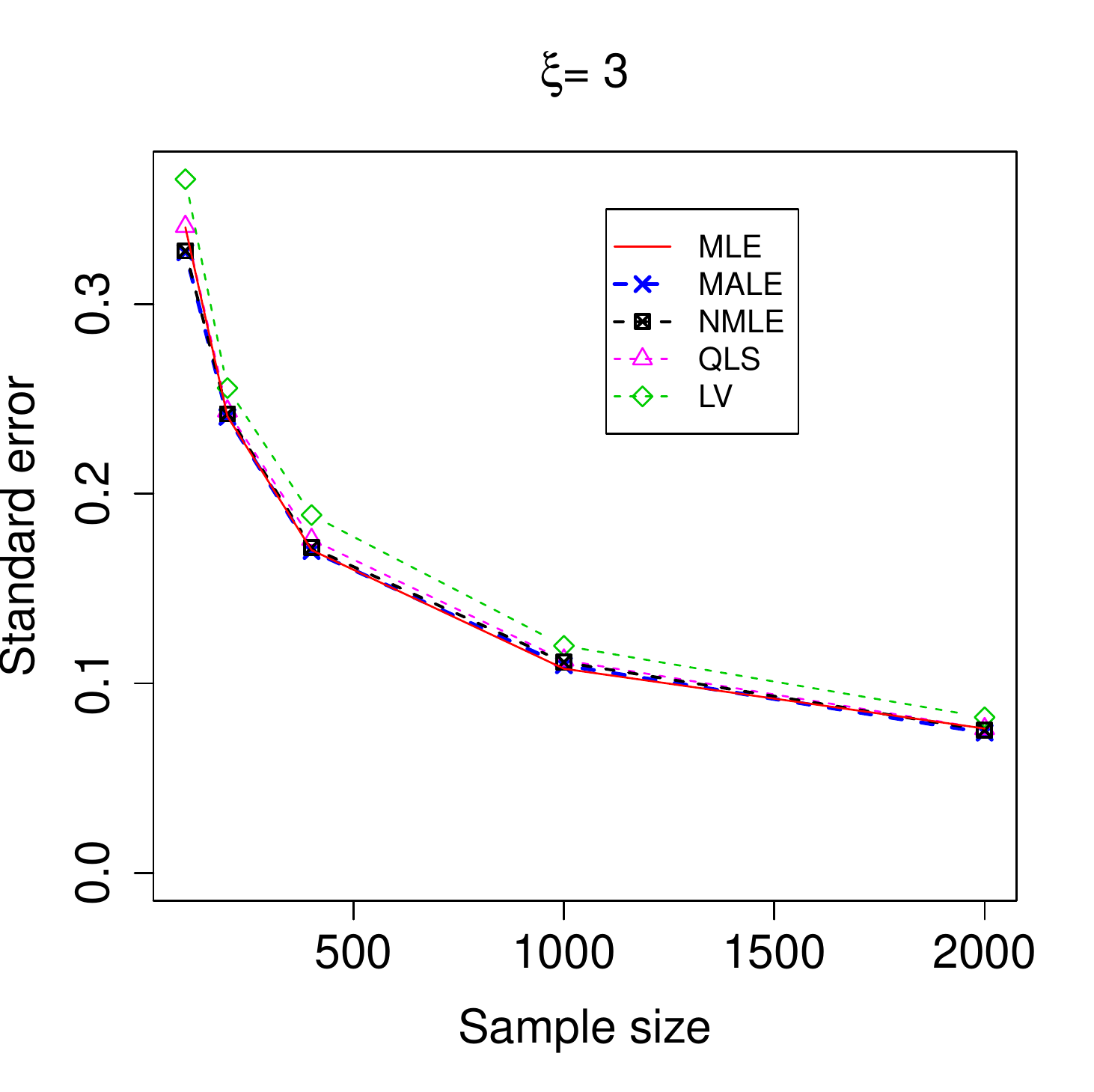} &
\includegraphics[width=0.40\textwidth]{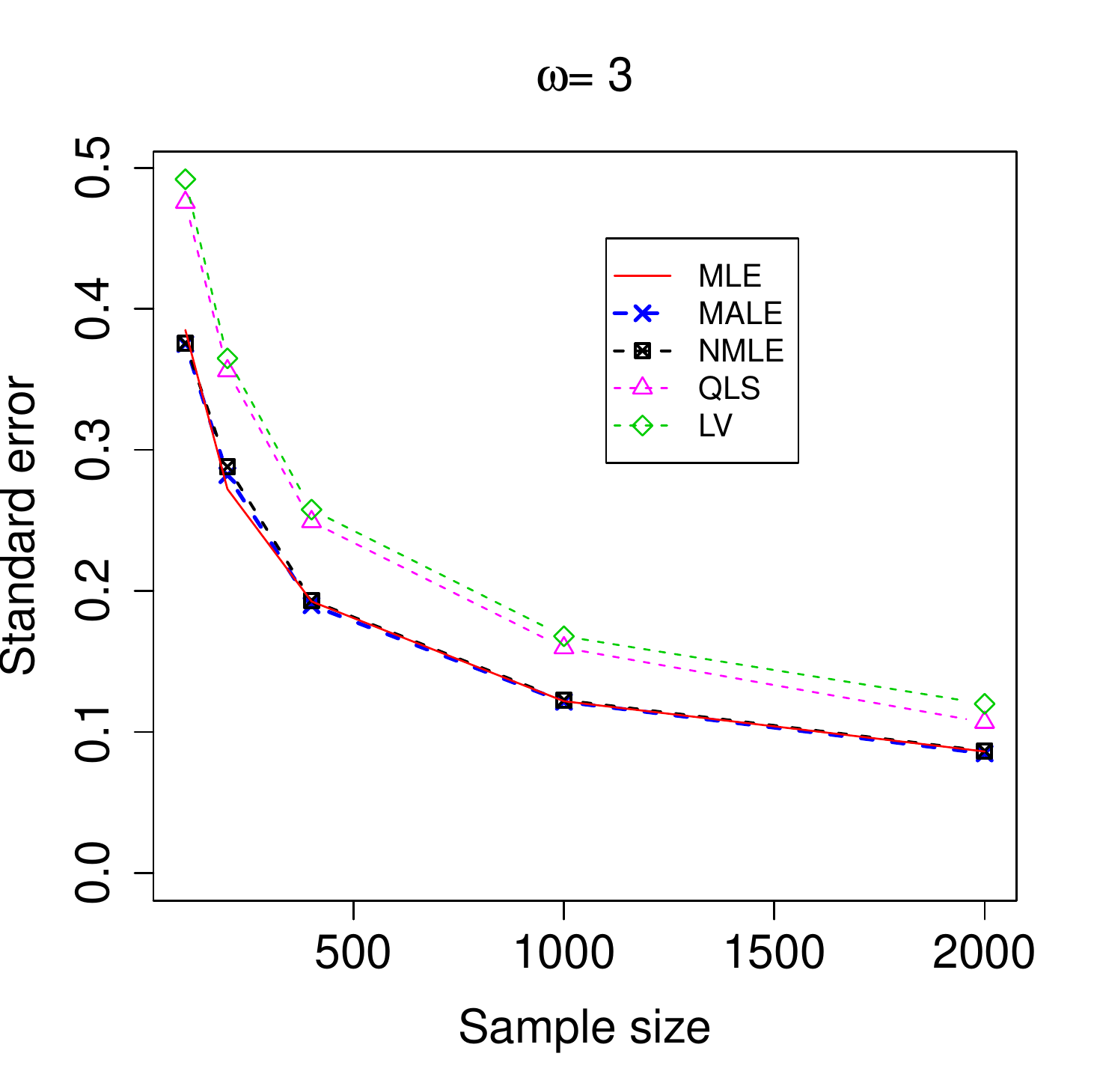}\\
\includegraphics[width=0.40\textwidth]{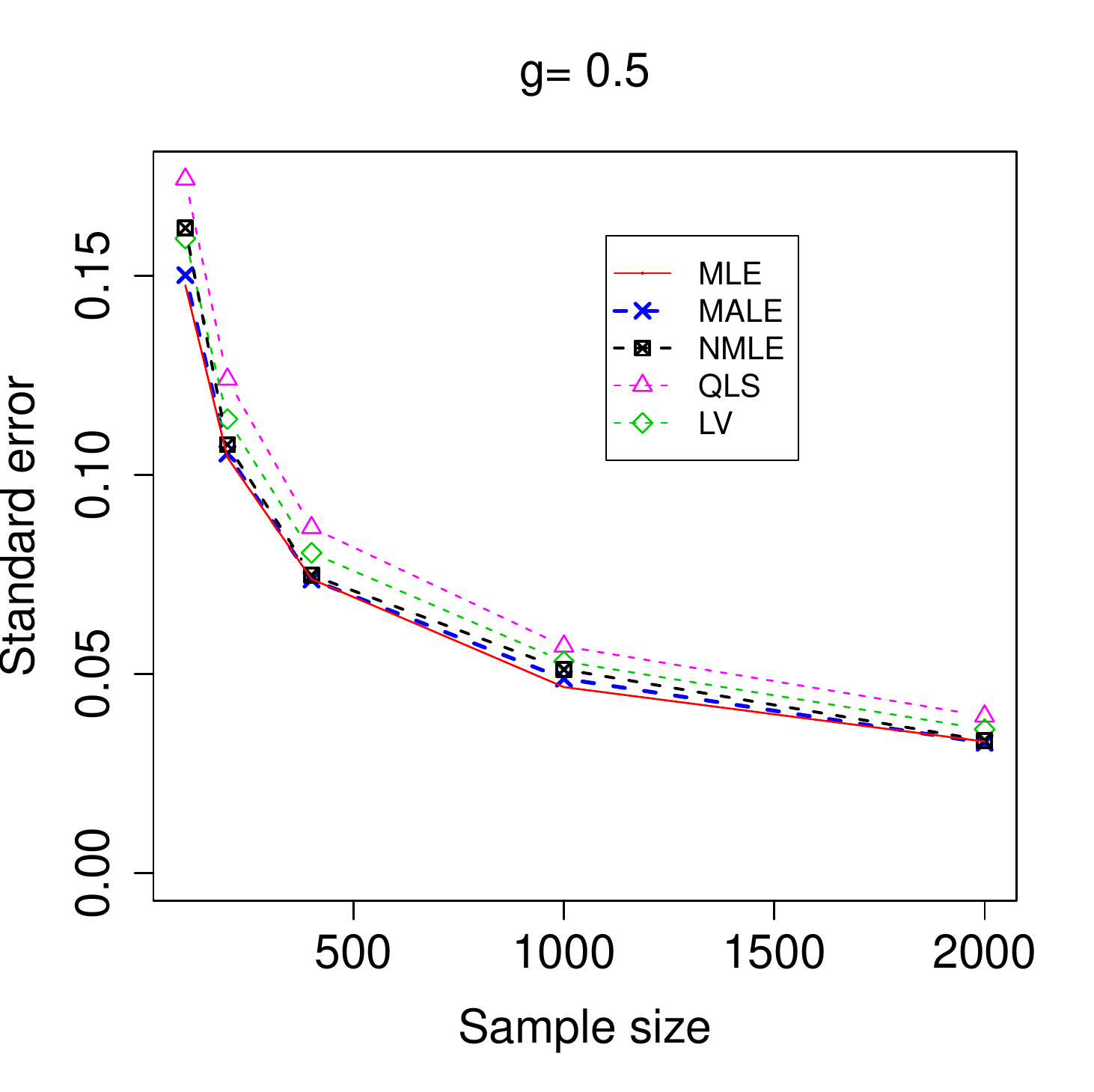} &
\includegraphics[width=0.40\textwidth]{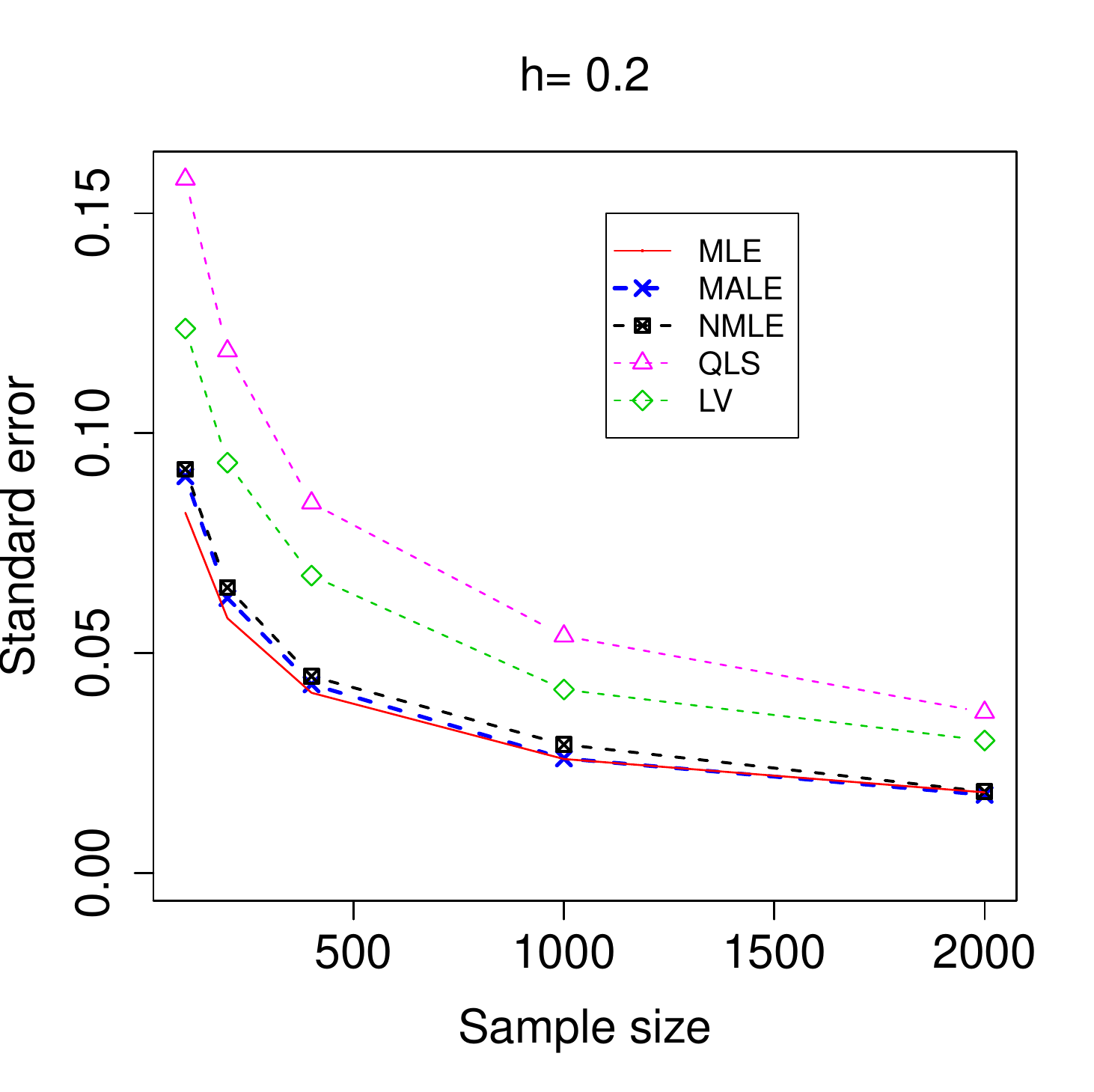}\\
\end{tabular}
\caption{The estimation efficiencies of four methods for Tukey's $g$-and-$h$ distribution; MLE: asymptotic standard error of the maximum likelihood estimator; MALE: maximum approximated likelihood estimator; NMLE: numerical maximum likelihood estimator;  QLS: quantile least square estimator; LV: letter-value-based estimator.}
\label{fig-2}
\end{figure}

Although the MALE and NMLE methods yield roughly the same estimation accuracy, the proposed MALE method is more computationally efficient. To show this, we record their CPU times (in seconds) from the above simulation study. The same initial values and optimization routine were used for both methods. All simulation runs were carried out in the software R on a cluster of 120 commodity Linux machines using a total of 300 CPU cores, with each core running at approximately 2 GFLOPS, or 2 billion floating point operations per second. The results based on the 1000 simulation runs are summarized in Table~\ref{tab-3}, where the first two columns are average CPU times (in seconds) for each simulation run for both methods and the third column is the average of the ratios of CPU times for these two methods applied to the same dataset. From Table~\ref{tab-3} we can see that, when the sample size grows to $n=2000$, our method is on average over 400 times faster than the NMLE method. We also would like to point out that the MALE method converges much faster than the NMLE method for almost all simulation runs, which also contributes to the smaller computation time of the proposed method.

\begin{table}[ht!]
 \caption{Comparisons of computational costs for MALE and NMLE}
 \centering
\begin{tabular}{cccc}
\hline
\hline
&\multicolumn{2}{c}{Average CPU time per run}&Average of CPU time ratio per run\\
\hline
sample size $n$&\hskip 2em MALE&NMLE&NMLE/MALE\\
\hline
100&\hskip 2em0.13&10.86&91.68\\
200&\hskip 2em0.14&22.60&173.80\\
400&\hskip 2em0.17&47.15&301.86\\
1000&\hskip 2em0.29&113.68&445.58\\
2000&\hskip 2em0.53&229.01&464.58\\
\hline
\hline
\end{tabular}
\label{tab-3}
\end{table}

\subsection{Evaluation of approximated likelihood ratio tests}

In the simulation example 2, we study the empirical size and the local power of the proposed approximated likelihood ratio tests. The null and alternative hypotheses are
\begin{enumerate}
\item $H_0^{(1)}: g=0$ vs $H_1: g=d/\sqrt{n}$ with $(d_1,d_2,d_3)=(0,1.5,3)$. The data were generated using model~(\ref{TGHdef}) with $(\xi_0,\omega_0,g_0,h_0)=(3,3,d/\sqrt{n},0.2)$;
\item $H_0^{(2)}: h=0$ vs $H_1: h=d/\sqrt{n}$ with $(d_1,d_2,d_3)=(0,0.5,1.0)$. The data were generated using model~(\ref{TGHdef}) with $(\xi_0,\omega_0,g_0,h_0)=(3,3,0.5,d/\sqrt{n})$;
\item $H_0^{(3)}: g=h=0$ vs $H_1: g=3d/\sqrt{n}, h=d/\sqrt{n}$ with $(d_1,d_2,d_3)=(0,0.5,1.0)$. The data were generated using model~(\ref{TGHdef}) with $(\xi_0,\omega_0,g_0,h_0)=(3,3,3d/\sqrt{n},d/\sqrt{n})$.
\end{enumerate}

To study the ability of the proposed tests in detecting small departures from the null hypotheses, we deliberately set up alternative hypotheses as a function of $\sqrt{n}$, which are commonly referred to as local alternatives in the literature. Studying the power of hypothesis tests under local alternatives is more accurate in characterizing asymptotic behavior of the test statistic and has been widely used, for example, see Zhang (2001) and Xu \& Wang (2010). The empirical powers based on $1000$ simulation runs are summarized in Table~{\ref{tab-1}. For $H_0^{(1)}$ and $H_0^{(3)}$, we see that under the null hypotheses, the empirical sizes are quite close to the nominal sizes. However, under the null hypothesis $H_0^{(2)}$, the empirical sizes are much smaller than the nominal level, even for relatively large $n$. In other words, in such a situation, the approximated likelihood ratio test tend to be conservative, which confirms our discussion in Section 3.2. However, all tests possess significant local powers in detecting very small deviations from the null hypotheses and these local powers increase when the departure from the null hypothesis grows from $d_2/\sqrt{n}$ to $d_3/\sqrt{n}$.
\begin{table}
 \caption{Empirical powers of the approximated likelihood ratio tests in the simulation example~2.}
 \centering
\begin{tabular}{*{2}{c}*{3}{c}rrrrrr}
\hline
           &       &   \multicolumn{3}{c}{$H_0:g=0$}&\multicolumn{3}{c}{$H_0: h=0$}&\multicolumn{3}{c}{$H_0: g=h=0$}\\
\hline
           &       &   \multicolumn{3}{c}{Nominal levels} &\multicolumn{3}{c}{Nominal levels}&\multicolumn{3}{c}{Nominal levels}\\
 $d$       &   $n$  &   10.0    &   5.0 &   1.0 &   10.0    &   5.0 &   1.0 &   10.0    &   5.0 &   1.0 \\
\hline
 $d_1$      &   $100$   &   12.3    &   6.3 &   1.8 &   4.0 &   2.0 &   0.4 &  7.7     & 4.5      & 1.1      \\
            &   $200$   &   11.4    &   5.7 &   1.4 &   4.5 &   2.2 &   0.1 &  9.6     & 4.7      & 0.8      \\
            &   $400$   &   10.4    &   5.0 &   1.6 &   5.8 &   2.4 &   0.5 &  8.7     & 4.0      & 0.4      \\
            &   $1000$   &   9.8    &   3.8 &   0.5 &   5.6 &   2.7 &  0.3 &  8.9     & 3.9      & 0.5      \\
            &   $2000$   &   8.8    &   4.1 &   1.3 &   7.1 &   3.4 &  0.8 &  8.9     & 4.6      & 0.7      \\
 \hline
 $d_2$      &   $100$   &   30.2    &   19.9 &   7.0 &   27.9 &   20.6 &   8.7 &  56.4     & 46.8      & 30.0      \\
            &   $200$   &   27.9    &   19.4 &   5.8 &   33.9 &   25.4 &   10.4 &  58.9     & 47.7      & 29.2      \\
            &   $400$   &   28.0    &   17.6 &   6.3 &   40.4 &   29.0 &   11.6 &  60.5     & 49.3      & 30.3      \\
            &   $1000$  &   28.5    &   19.0 &   7.2 &   47.8 &   33.1 &  15.6 &  64.6     & 53.2      & 32.5      \\
            &   $2000$   &  28.1    &   18.4 &   7.9 &   53.7 &   38.8 &  17.8 &  65.0     & 53.6      & 30.7      \\
 \hline
 $d_3$      &   $100$   &   67.4    &   55.9 &   30.8 &   58.0 &   48.7 &   30.3 &  93.0     & 88.6      & 78.4      \\
            &   $200$   &   68.6    &   57.3 &   32.6 &   68.2 &   58.7 &   41.7 &  95.3     & 93.3      & 84.0      \\
            &   $400$   &   67.7    &   57.9 &   35.9 &   79.5 &   69.7 &   49.8 &  97.2     & 95.2      & 88.4      \\
            &   $1000$  &   70.9    &   59.8 &   35.2 &   83.2 &   74.8 &  58.9 &  98.5     & 96.6      & 91.7      \\
            &   $2000$   &  68.9    &   58.0 &   34.9 &   89.5 &   83.3 &  63.9 &  98.8     & 97.5      & 92.1      \\
 \hline

 \end{tabular}
\label{tab-1}
\end{table}

\section{Air Pollution Data}

In this section, we apply the proposed method to two years of air pollution data obtained from Ledolter \& Hogg (2009, Chapter 1, Section 1.4), which can be downloaded at the website http://www.biz.uiowa.edu/faculty/jledolter/AppliedStatistics/. In each of 1976 and 1977, $64$  weekday afternoon lead concentrations, in micrograms/m$^3$, were collected near the San Diego Freeway in Los Angeles. The same datasets were also used by Rayner \& MacGillivray (2002a, 2002b). In Figure~\ref{fig-4}, the histogram and the density plots show possible skewness and some outliers in both datasets. Following Rayner \& MacGillivray (2002b), we use Tukey's $g$-and-$h$ distribution to fit both datasets.

\begin{figure}[ht!]
\centering
\begin{tabular}{cc}
\includegraphics[width=0.45\textwidth]{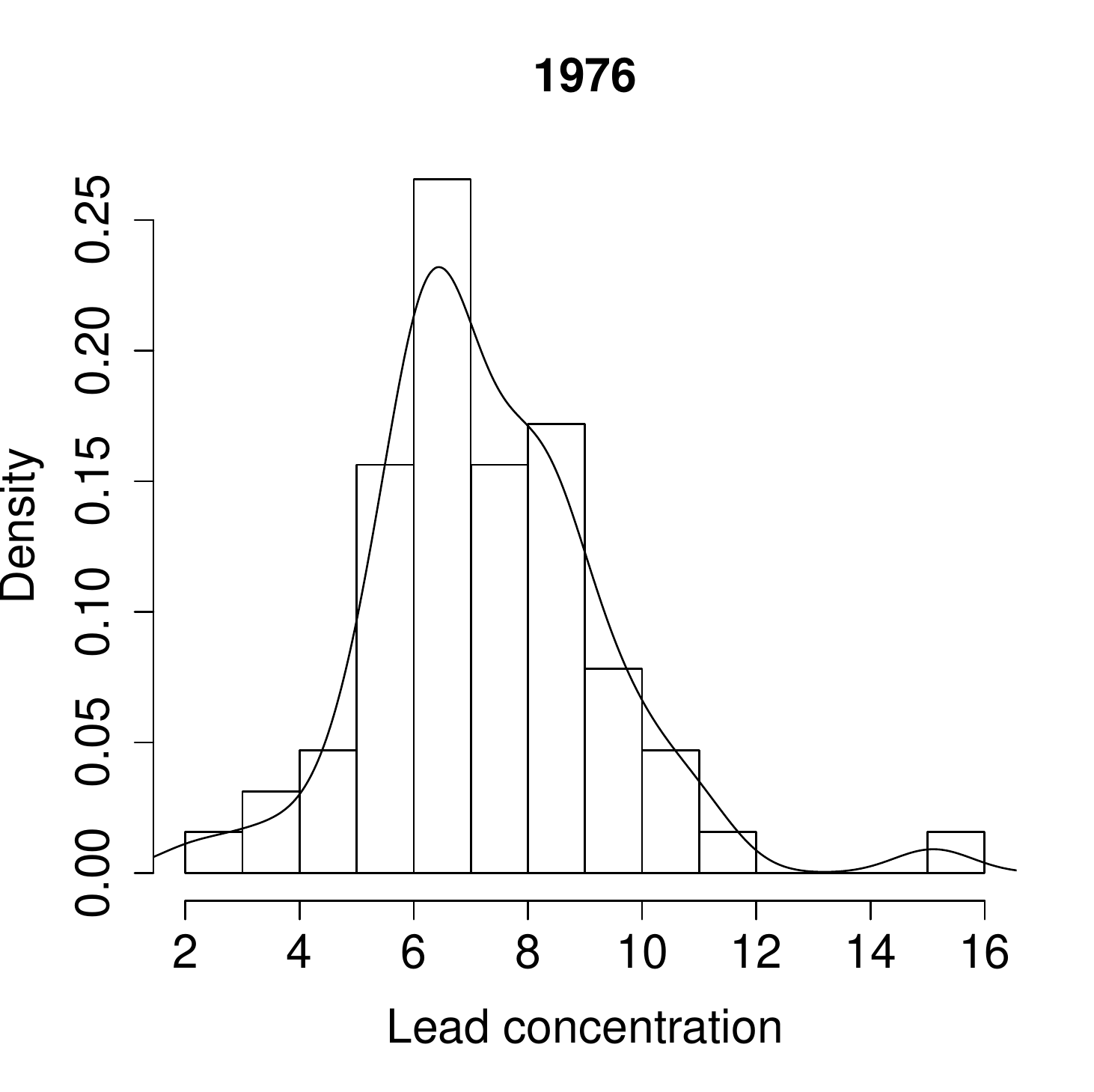} &
\includegraphics[width=0.45\textwidth]{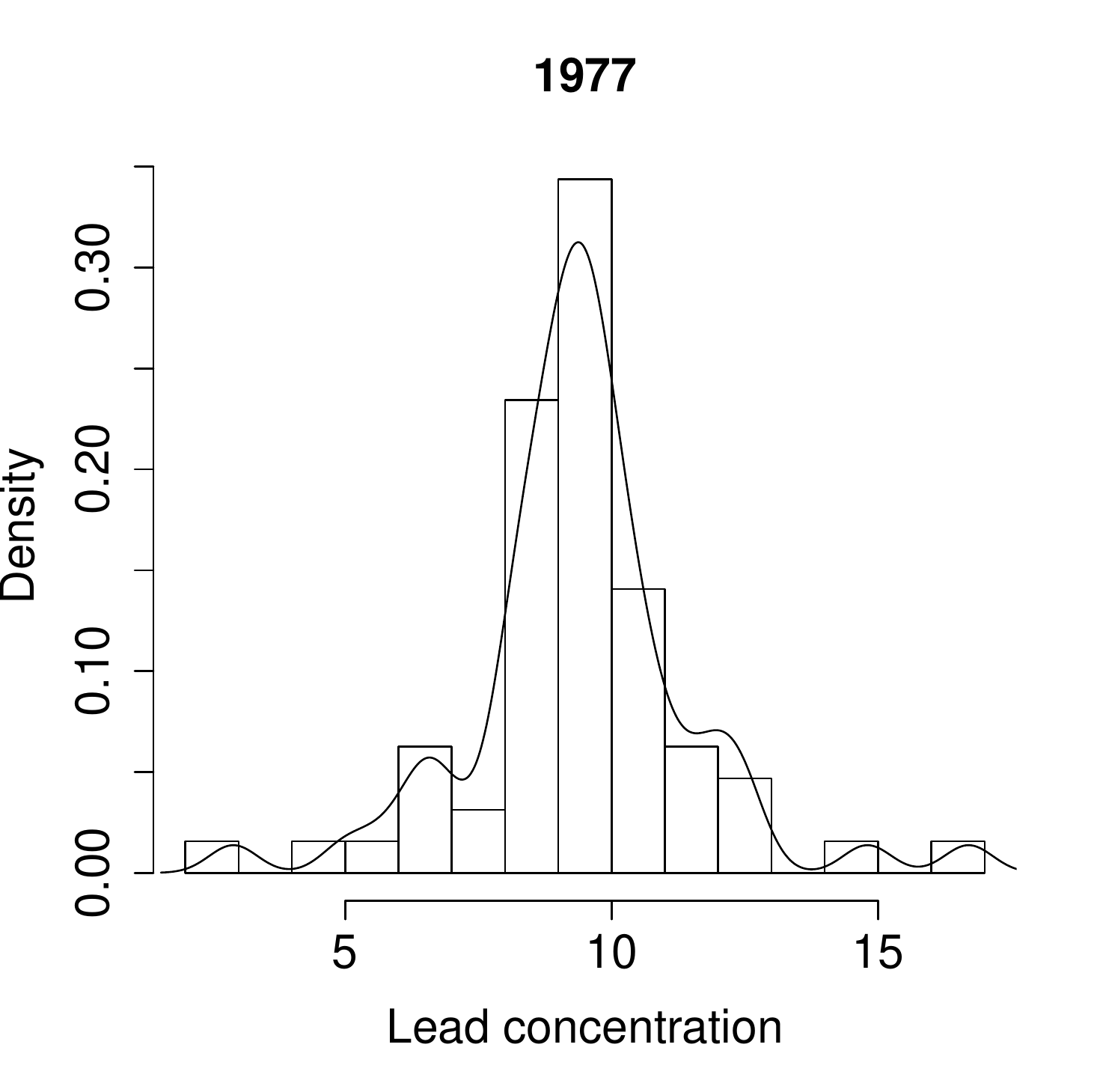}
\end{tabular}
\caption{Histograms of lead concentrations in 1976 and 1977.}
\label{fig-4}
\end{figure}

For parameter estimation, the four methods described in Section 4.1 were used. All estimates are summarized in Table~\ref{tab-2}. As we can see, the estimates $\hat{\btheta}_{male,n}$ and $\hat{\btheta}_{nmle,n}$ are identical. This is not surprising because both of these methods should yield estimators that are very close to the true maximum likelihood estimator. Although $\hat{\btheta}_{lv,n}$ and $\hat{\btheta}_{qls,n}$ yield consistent estimates of $\xi$, $\omega$ and $g$, compared to $\hat{\btheta}_{male,n}$ and $\hat{\btheta}_{nmle,n}$, they produce very different estimates for $h$. For example, for 1977, $\hat{\btheta}_{qls,n}$ gives $\hat h_{qls,n}=0.49$, which is very questionable because with this value the fitted distribution basically has an infinite variance.

We also ran a series of hypothesis tests using the approximated likelihood ratio test statistics and the results are summarized in Table~\ref{tab-2}, where $C_{0.95}$ stands for the critical value derived from the asymptotic distribution given in Theorem~\ref{thm2} for each test at the $0.05$ significance level. Based on these tests, there is significant evidence that $h_0$ is not $0$ but it seems reasonable to set $g=0$ for both datasets. This indicates that for both datasets, the data do not have a normal distribution but a heavy-tailed symmetric distribution.

\begin{table}[ht!]
 \caption{Estimates and hypothesis tests of lead concentration data.}
 \centering
\begin{tabular}{llll}
\hline
&Methods& Year 1976&Year 1977\\
\hline
&$\hat{\btheta}_{male,n}$&$(7.11,1.60,0.19,0.12)$&$(9.38,1.22 ,0.04, 0.34)$\\
Estimation&$\hat{\btheta}_{nmle,n}$&$( 7.11, 1.60, 0.19 ,0.12)$&$(9.38, 1.22, 0.04, 0.34)$\\
$(\hat\xi,\hat\omega,\hat g,\hat h)$&$\hat{\btheta}_{lv,n}$&$(6.95, 1.62, 0.25, 0.07)$&$(9.40,1.33, 0.04, 0.21)$\\
&$\hat{\btheta}_{qls,n}$&$(7.12,1.62, 0.24, 0.01)$&$(9.41, 1.08, 0.01, 0.49)$\\
\hline
Hypothesis&$H_0$: $g=h=0$&$(7.26,5.14)$, Reject;&$(14.34, 5.14)$, Reject;\\
Tests&$H_0$: $h=0$&$(4.46, 2.71)$, Reject;& $(14.05, 2.71)$, Reject;\\
$(D_n,C_{0.95})$&$H_0$: $g=0$&$(1.54, 3.84)$, Fail to reject.&$(0.03, 3.84)$, Fail to reject.\\
\hline
\end{tabular}
\label{tab-2}
\end{table}

\section{Conclusion}

In this paper, we proposed a computationally efficient estimation method for Tukey's $g$-and-$h$ distribution. The proposed maximum approximated likelihood estimator has the same limiting distribution as the true maximum likelihood estimator and the resulting approximated likelihood ratio tests can be used to test interesting hypotheses involving the two shape parameters $g$ and $h$. The performances of the proposed method have been demonstrated through extensive simulation studies and an application to a air pollution data.

There are several ways to extend the current results. Tukey's $g$-and-$h$ distribution can be used as the residual process of a regression model to conduct robust model selection and model averaging (Xu et al., 2014). Another interesting application is to conduct model selection for an autoregressive model with an infinite variance as in Xu et al. (2012), where Tukey's $g$-and-$h$ distribution can be used to replace the stable distribution as the innovation process. It would also be interesting to conduct some theoretical investigations on the local powers of the ALRT tests in Section 3 and how to construct valid confidence intervals for $g$ and $h$, without knowing the true value $h_0=0$ or not.

\baselineskip=23pt

\section*{Appendix: Technical Proofs}
\setcounter{lem}{0}

\setcounter{equation}{0}

\renewcommand{\theequation}{A.\arabic{equation}}
\renewcommand{\thelem}{A.\arabic{lem}}

We first define some functions $C_1(z,\btheta)=\frac{\partial \varphi_{\btheta}(z_p)}{\partial^{\dag}z_p }\large|_{z_p=z}$, $C_2(z,\btheta)=\frac{\partial^2 \varphi_{\btheta}(z_p)}{\partial\btheta \partial^{\dag}z_p }\large|_{z_p=z}$, $C_3(z,\btheta)=\frac{\tau_{g,h}''(z)}{\tau_{g,h}'(z)}\frac{\partial \varphi_{\btheta}(z_p)}{\partial^{\dag} z_p}\frac{\partial z_p}{\partial \btheta}\large |_{z_p=z}$, $C_4(z,\btheta)=\frac{1}{K_n}\frac{\partial \log\tau_{g,h}'(z)}{\partial\btheta}\frac{\partial \varphi_{\btheta}(z_p)}{\partial^{\dag} z_p}\large |_{z_p=z}$, $C_5(z,\btheta)=\frac{\partial \varphi_{\btheta}(z_p)}{\partial^{\dag} z_p}\frac{\partial^2 z_p}{\partial \btheta\partial^{\dag}z_p}\large |_{z_p=z}$. We list some technical conditions needed for proofs of Theorem~\ref{thm1}-\ref{thm2}, where Conditions 2-3 are taken from Self \& Liang (1987).\\
{\it Condition 1}: There exists an open neighborhood of $\btheta_0$, denoted by $\Theta_{0,n}\subseteq \Bbb{R}^4$, such that $\quad \textsf{Y}_{1,\btheta}\leq y_{\min}<y_{\max}\leq \textsf{Y}_{K_n,\btheta}$ with probability 1 for all $\btheta\in\Theta_{0,n}$;\\
{\it Condition 2}: The first three derivatives of $L_n(\btheta)$ with respect to $\btheta$ exist with probability $1$ on $\Theta_{0,n}\cap \Omega$;\\
{\it Condition 3}: The third derivative of $n^{-1}L_n(\btheta)$ is bounded by some function $\psi_n(\bY)$ with $|E\{\psi_n(\bY)\}|<\infty$;\\
{\it Condition 4}: $I(\btheta)=E\{I_{n}(\btheta)\}$ is positive definite on $\Theta_{0,n}$ and $I(\btheta_0)=\mbox{var}\left\{n^{-1/2}U_n(\btheta_0)\right\}$;\\
{\it Condition 5}: $\frac{2\sqrt{n}b_n}{K_n}\sup_{\btheta\in\Theta_{0,n}\cap\Omega,z\in [-b_n,b_n]}\left|C_j(z,\btheta)\right|\to0$ for $j=1,\dots,5$ as $n\to\infty$ and $K_n\to\infty$.
\begin{lem}
\label{lema1}
Under Conditions 1-5, as $K_n\to\infty$ and $n\to\infty$, we have that
\begin{eqnarray}
\frac{1}{\sqrt{n}}\sup_{\btheta\in\Theta_{0,n}\cap\Omega}|\tilde{L}_n(\btheta)-L_n(\btheta)|&=&o_p(1),\label{eq1:lema1} \\
\frac{1}{\sqrt{n}}\sup_{\btheta\in\Theta_{0,n}\cap\Omega}|\tilde{U}_n(\btheta)-U_n(\btheta)|&=&o_p(1).\label{eq2:lema1}
\end{eqnarray}
\end{lem}

\begin{proof} For a fixed $\btheta\in\Theta_{0,n}\cap\Omega$, let $z_p=\tau_{g,h}^{-1}(\frac{y-\xi}{\omega})$ and suppose $\textsf{Z}_k\leq z_p <\textsf{Z}_{k+1}$ for some $1\leq k\leq K_n$. Then we can define the corresponding $\tilde{z}_{p,k}$ as in (\ref{zhat}) such that $|z_p-\tilde{z}_{p,k}|\leq 2b_n/K_n$, which implies that
\[
\left|\varphi_{\btheta}(z_p)-\tilde{\varphi}_{\btheta}(z_p)\right|\leq \frac{2b_n}{K_n}\sup_{z_p\in [-b_n,b_n]}\left|\frac{\partial \varphi_{\btheta}}{\partial^{\dag} z_p}(z_p)\right|.
\]
Then, it is straightforward to show that, under Conditions 1 and 5, we have
\[
\begin{split}
\frac{1}{\sqrt{n}}\sup_{\btheta\in\Theta_{0,n}\cap\Omega}|\tilde{L}_n(\btheta)-L_n(\btheta)|&\leq \sup_{\btheta\in\Theta_{0,n}\cap\Omega}\frac{1}{\sqrt{n}}\sum_{i=1}^n|\varphi_{\btheta}(z_{p_i})-\tilde{\varphi}_{\btheta}(z_{p_i})|
\\&\leq \frac{2\sqrt{n}b_n}{K_n}\sup_{\btheta\in\Theta_{0,n}\cap\Omega,z_p\in [-b_n,b_n]}|C_1(z_p,\btheta)|=o_p(1).
\end{split}
\]
By definition, we can rewrite (\ref{zhat}) as
\[
\tilde{z}_{p,k}=\textsf{Z}_k+\frac{y-\textsf{Y}_{k,\btheta}}{\textsf{Y}_{k+1,\btheta}-\textsf{Y}_{k,\btheta}}(\textsf{Z}_{k+1}-\textsf{Z}_k)=\textsf{Z}_k+\frac{\tau_{g,h}(z_p)-\tau_{g,h}(\textsf{Z}_k)}{\tau_{g,h}(\textsf{Z}_{k+1})-\tau_{g,h}(\textsf{Z}_k)}(\textsf{Z}_{k+1}-\textsf{Z}_k).
\]
Applying the chain rule of differentiation, we have
\[
u_{\btheta}(z_p)=\frac{\partial \varphi_{\btheta}(z_p)}{\partial \btheta}=\frac{\partial \varphi_{\btheta}(z_p)}{\partial^{\dag} \btheta}+\frac{\partial \varphi_{\btheta}(z_p)}{\partial^{\dag} z_p}\frac{\partial z_p}{\partial \btheta}=v_{\btheta}(z_p)+w_{\btheta}(z_p)x_{\btheta}(z_p),
\]
where we denote $v_{\btheta}(z_p)=\frac{\partial \varphi_{\btheta}(z_p)}{\partial^{\dag} \btheta}$, $w_{\btheta}(z_p)=\frac{\partial \varphi_{\btheta}(z_p)}{\partial^{\dag} z_p}$ and $x_{\btheta}(z_p)=\frac{\partial z_p}{\partial \btheta}$. Similarly, we have
\[
\tilde{u}_{\btheta}(z_p)=\frac{\partial \varphi_{\btheta}(\tilde{z}_{p,k})}{\partial \btheta}=\frac{\partial \varphi_{\btheta}(\tilde{z}_{p,k})}{\partial^{\dag} \btheta}+\frac{\partial \varphi_{\btheta}(\tilde{z}_{p,k})}{\partial^{\dag} \tilde{z}_{p,k}}\frac{\partial \tilde{z}_{p,k}}{\partial \btheta}=v_{\btheta}(\tilde{z}_{p,k})+w_{\btheta}(\tilde{z}_{p,k})\tilde{x}_{\btheta}(z_p),
\]
where we denote $\tilde{x}_{\btheta}(z_p)=\frac{\partial \tilde{z}_{p,k}}{\partial \btheta}$.
Straightforward algebra yields
\be
\label{a1}
\begin{split}
u_{\btheta}(z_p)-\tilde{u}_{\btheta}(z_p)&=u_{\btheta}(z_p)-u_{\btheta}(\tilde{z}_{p,k})+w_{\btheta}(\tilde{z}_{p,k})\left\{x_{\btheta}(\tilde{z}_{p,k})-\tilde{x}_{\btheta}(z_p)\right\}\\
&=\frac{\partial^2 \varphi_{\btheta}}{\partial\btheta \partial^{\dag}z_p }(z_{p,1}^*)(z_p-\tilde{z}_{p,k})+w_{\btheta}(\tilde{z}_{p,k})\left\{x_{\btheta}(\tilde{z}_{p,k})-\tilde{x}_{\btheta}(z_p)\right\},
\end{split}
\ee
where $z_{p,1}^*\in [\textsf{Z}_{k}, \textsf{Z}_{k+1})$. Furthermore, we can show that
\be
\label{a2}
\begin{split}
\tilde{x}_{\btheta}(z_p)-x_{\btheta}(\tilde{z}_{p,k})&=\tilde{x}_{\btheta}(z_p)-x_{\btheta}(z_p)+x_{\btheta}(z_p)-x_{\btheta}(\tilde{z}_{p,k})=\frac{\partial \tilde{z}_{p,k}}{\partial \btheta}-x_{\btheta}(z_p)+x_{\btheta}(z_p)-x_{\btheta}(\tilde{z}_{p,k})\\&=\frac{\partial \tilde{z}_{p,k}}{\partial^{\dag} z_p}\frac{\partial z_p}{\partial \btheta}+\frac{\partial \tilde{z}_{p,k}}{\partial^{\dag} \btheta}-x_{\btheta}(z_p)+x_{\btheta}(z_p)-x_{\btheta}(\tilde{z}_{p,k})\\&=\left\{\frac{\tau_{g,h}'(z_p)}{\tau_{g,h}(\textsf{Z}_{k+1})-\tau_{g,h}(\textsf{Z}_k)}(\textsf{Z}_{k+1}-\textsf{Z}_k)-1\right\}\frac{\partial z_p}{\partial \btheta}+\frac{\partial \tilde{z}_{p,k}}{\partial^{\dag} \btheta}+x_{\btheta}(z_p)-x_{\btheta}(\tilde{z}_{p,k})
\\&=\frac{\tau_{g,h}''(z_{p,3}^*)(z_{p}-z_{p,2}^*)}{\tau_{g,h}'(z_{p,2}^*)}\frac{\partial z_p}{\partial \btheta}+\frac{\frac{\partial \tau_{g,h}'}{\partial\btheta}(z_{p,4}^*)(z_{p,5}^*-z_{p,6}^*)}{\tau_{g,h}'(z_{p,2}^*)K_n}+\frac{\partial^2 z_p}{\partial \btheta\partial^{\dag}z_p}(z_{p,7}^*)(z_p-\tilde{z}_{p,k}),
\end{split}
\ee
where $\textsf{Z}_k\leq z_{p,i}^*\leq \textsf{Z}_{k+1}$ with $2\leq i\leq 7$. Then, combining~(\ref{a1}) and (\ref{a2}), under Condition~5 and the fact that $\textsf{Z}_{k+1}-\textsf{Z}_k=2b_n/K_n$, we have
\[
\sup_{\btheta\in\Theta_{0}\cap\Omega,z_p\in [-b_n,b_n]}\sqrt{n}|u_{\btheta}(z_p)-\tilde{u}_{\btheta}(z_p)|=o_p(1),
\]
which further leads to, under Conditions~1 and 5, that
\[
\begin{split}
\frac{1}{\sqrt{n}}\sup_{\btheta\in\Theta_{0}\cap\Omega}|\tilde{U}_n(\btheta)-U_n(\btheta)|&\leq \sup_{\btheta\in\Theta_{0,n}\cap\Omega}\frac{1}{\sqrt{n}}\sum_{i=1}^n|u_{\btheta}(z_{p_i})-\tilde{u}_{\btheta}(z_{p_i})|
\\&\leq \sup_{\btheta\in\Theta_{0,n}\cap\Omega,z_p\in [-b_n,b_n]}\sqrt{n}|u_{\btheta}(z_{p})-\tilde{u}_{\btheta}(z_{p})|=o_p(1).
\end{split}
\]
\end{proof}

Before proceeding to proofs of Theorem~\ref{thm1}-\ref{thm2}, we first define quantities $\zeta_n=n^{-1}I^{-1}(\btheta_0)U_n(\btheta_0)$ and $H_n(\btheta)=-\{\zeta_n-(\btheta-\btheta_0)\}^TI(\btheta_0)\{\zeta_n-(\btheta-\btheta_0)\}+n^{-2}U_n^T(\btheta_0)^TI^{-1}(\btheta_0)U_n^T(\btheta_0)$.

\begin{proof}[Proof of Theorem~\ref{thm1}]
Under Conditions 2--4, by Lemma~1 in Self \& Liang (1987), for any $\btheta$ such that $\sqrt{n}(\btheta-\btheta_0)=O(1)$, we have
\be
\label{a3}
\frac{2}{n}\{L_n(\btheta)-L_n(\btheta_0)\}=H_n(\btheta)+R_n(\btheta),
\ee
where $R_n(\btheta)=O_p(1)||\btheta-\btheta_0||^3$.

Following the same arguments in the proof of Theorem 1 in Self \& Liang (1987) and using the equation~(\ref{eq1:lema1}) in Lemma~{\ref{lema1}}, we can show that there exists a sequence $\hat{\btheta}_{male,n}\in\Omega$ such that $|\hat{\btheta}_{male,n}-\btheta_0|=O_p(1/\sqrt{n})$. Denote by $\hat{\btheta}_{H_n}$ the maximizer of the quadratic function $H_n(\btheta)$. Because of the convexity of $\Omega$, it is straightforward to show that $|\hat{\btheta}_{H_n}-\btheta_0|=O_p(1/\sqrt{n})$ and hence $|\hat{\btheta}_{male,n}-\hat{\btheta}_{H_n}|=O_p(1/\sqrt{n})$. Therefore, we must have that
\be
\label{a4}
\begin{split}
\frac{2}{n}&\{\tilde{L}_n(\hat{\btheta}_{male,n})-\tilde{L}_n(\hat{\btheta}_{H_n})\}\\&=\frac{2}{n}\{\tilde{L}_n(\hat{\btheta}_{male,n})-\tilde{L}_n(\hat{\btheta}_{H_n})\}-\frac{2}{n}\{L_n(\hat{\btheta}_{male,n})-L_n(\hat{\btheta}_{H_n})\}+\frac{2}{n}\{L_n(\hat{\btheta}_{male})-L_n(\hat{\btheta}_{H})\}
\\&=\frac{2}{n}\{\tilde{L}_n(\hat{\btheta}_{male,n})-L_n(\hat{\btheta}_{male,n})\}-\frac{2}{n}\{\tilde{L}_n(\hat{\btheta}_{H_n})-L_n(\hat{\btheta}_{H_n})\}+\frac{2}{n}\{L_n(\hat{\btheta}_{male,n})-L_n(\hat{\btheta}_{H_n})\}
\\&=\frac{2}{n}\{\tilde{U}_n(\hat{\btheta}_1^*)-U_n(\hat{\btheta}_1^*)\}(\hat{\btheta}_{male,n}-\hat{\btheta}_{H_n})+H_n(\hat{\btheta}_{male,n})-H_n(\hat{\btheta}_{H_n})+R_n(\hat{\btheta}_{male,n})-R_n(\hat{\btheta}_{H_n})\\
&=H_n(\hat{\btheta}_{male,n})-H_n(\hat{\btheta}_{H_n})+o_p(1/n),
\end{split}
\ee
where $\hat{\btheta}_1^*\in (\hat{\btheta}_{male,n}, \hat{\btheta}_{H_n})$. The last equation follows by combining~(\ref{a3}) and (\ref{eq2:lema1}) in Lemma~\ref{lema1}.

Comparing the left-hand side and right-hand side of~(\ref{a4}), by definitions of $\hat{\btheta}_{male,n}$ and $\hat{\btheta}_{H_n}$, we must have $\tilde{L}_n(\hat{\btheta}_{male,n})-\tilde{L}_n(\hat{\btheta}_{H_n})\geq0$ and $H_n(\hat{\btheta}_{male,n})-H_n(\hat{\btheta}_{H_n})\leq 0$, which implies that $|H_n(\hat{\btheta}_{male,n})-H_n(\hat{\btheta}_{H_n})|=o_p(1/n)$. Furthermore, since $H_n(\btheta)$ is a quadratic function, we conclude that $|\hat{\btheta}_{male,n}-\hat{\btheta}_{H_n}|=o_p(1/\sqrt{n})$.

The rest of the proof follows from Theorem~2 in Self \& Liang (1987) and its application to cases 1 and 2 in that paper.
\end{proof}

\begin{proof}[Proof of Theorem~\ref{thm2}]
Denote $\hat{\btheta}_{male,n}^{0}=\arg \sup_{\btheta\in\Omega_0}\tilde{L}_n(\btheta)$ and $\hat{\btheta}_{male,n}=\arg \sup_{\btheta\in\Omega}\tilde{L}_n(\btheta)$. Similarly, define the true maximum likelihood estimator as $\hat{\btheta}_{mle,n}^{0}=\arg \sup_{\btheta\in\Omega_0}L_n(\btheta)$ and $\hat{\btheta}_{mle,n}=\arg \sup_{\btheta\in\Omega}L_n(\btheta)$. Then it is easy to show that
\be
\label{a5}
\begin{split}
D_n&=-2\{\tilde{L}_n(\hat{\btheta}_{male,n}^0)-\tilde{L}_n(\hat{\btheta}_{male,n})\}\\
&=\underbrace{-2\{\tilde{L}_n(\hat{\btheta}_{male,n}^0)-\tilde{L}_n(\hat{\btheta}_{male,n})\}+2\{ L_n(\hat{\btheta}_{male,n}^0)- L_n(\hat{\btheta}_{male,n})\}}_{I}\\
&\qquad\underbrace{-2\{ L_n(\hat{\btheta}_{male,n}^0)- L_n(\hat{\btheta}_{male,n})\}+2\{ L_n(\hat{\btheta}_{mle,n}^0)- L_n(\hat{\btheta}_{mle,n})\}}_{II}\\
&\qquad\underbrace{-2\{ L_n(\hat{\btheta}_{mle,n}^0)- L_n(\hat{\btheta}_{mle,n})\}}_{D_n^*}.
\end{split}
\ee
Under the null hypothesis $H_0: \btheta\in\Omega_0$, we have shown in the proof of Theorem~\ref{thm1} that $|\hat{\btheta}_{male,n}^0-\btheta_0|=O_p(1/\sqrt{n})$ and $|\hat{\btheta}_{male,n}-\btheta_0|=O_p(1/\sqrt{n})$. Hence, $|\hat{\btheta}_{male,n}^0-\hat{\btheta}_{male,n}|=O_p(1/\sqrt{n})$.  Using~(\ref{eq2:lema1}) in Lemma~\ref{lema1}, we can show that
\[
\begin{split}
I&=-2\{\tilde{L}_n(\hat{\btheta}_{male,n}^0)-L_n(\hat{\btheta}_{male,n}^0)\}+2\{ \tilde{L}_n(\hat{\btheta}_{male,n})- L_n(\hat{\btheta}_{male,n})\}\\
&=-2\{\tilde{U}_n(\hat{\btheta}_1^*)-U_n(\hat{\btheta}_1^*)\}(\hat{\btheta}_{male,n}^0-\hat{\btheta}_{male,n})=o_p(1),
\end{split}
\]
where $\hat{\btheta}_1^*\in(\hat{\btheta}_{male,n}^0, \hat{\btheta}_{male,n})$.

Lemma~1 in Self \& Liang (1987) showed that, under the null hypothesis, $H_0: \btheta\in\Omega_0$, one has $|\hat{\btheta}_{mle,n}^0-\hat{\btheta}_{H_n}^0|=o_p(1/\sqrt{n})$ and $|\hat{\btheta}_{mle,n}-\hat{\btheta}_{H_n}|=o_p(1/\sqrt{n})$. Following the proof of Theorem~\ref{thm1}, we can also show that $|\hat{\btheta}_{male,n}^0-\hat{\btheta}_{H_n}^0|=o_p(1/\sqrt{n})$ and $|\hat{\btheta}_{male,n}-\hat{\btheta}_{H_n}|=o_p(1/\sqrt{n})$. Therefore $|\hat{\btheta}_{male,n}^0-\hat{\btheta}_{mle,n}^0|=o_p(1/\sqrt{n})$ and  $|\hat{\btheta}_{male,n}-\hat{\btheta}_{mle,n}|=o_p(1/\sqrt{n})$. Then, it is easy to show that
\[
\begin{split}
II&=-2\{L_n(\hat{\btheta}_{male,n}^0)-L_n(\hat{\btheta}_{mle,n}^0)\}+2\{ L_n(\hat{\btheta}_{male,n})- L_n(\hat{\btheta}_{mle,n})\}\\
&=-2U_n^0(\hat{\btheta}_2^*)(\hat{\btheta}_{male,n}^0-\hat{\btheta}_{mle,n}^0)+2U_n(\hat{\btheta}_3^*)(\hat{\btheta}_{male,n}-\hat{\btheta}_{mle,n})=o_p(1),
\end{split}
\]
where $\hat{\btheta}_2^*\in (\hat{\btheta}_{male,n}^0,\hat{\btheta}_{mle,n}^0)$ and $\hat{\btheta}_3^*\in (\hat{\btheta}_{male,n},\hat{\btheta}_{mle,n})$ and $U_n^0(\hat{\btheta}_2^*)$ is the partial derivative function of $L_n(\btheta)$ in the restricted parameter space $\Omega_0$. The last equation follows from the fact $1/\sqrt{n}U_n^0(\btheta)$ and $\/\sqrt{n}U_n(\btheta)$ are bounded for $\btheta$ in a ``small" neighborhood of $\hat{\btheta}_{mle,n}^0$ and $\hat{\btheta}_{mle,n}$, respectively.

By plugging $I$ and $II$ back into equation~(\ref{a3}), we have $D_n=D_n^*+o_p(1)$, which implies that $D_n$ and $D_n^*$ have the same asymptotic distribution. The distribution of $D_n^*$ was derived in Self \& Liang (1987). The rest of the proof follows from Theorem 3 in Self \& Liang (1987) as well as its application to cases 4-6 in that paper.
\end{proof}

\pagestyle{empty}
\section*{References}
\baselineskip 15pt

\refmark Bolker, B. M. (2008),
{\it Ecological Models and Data in R}, Princeton University Press.

\refmark Byrd, R. H., Lu, P., Nocedal, J. \& Zhu, C. (1995), ``A limited memory algorithm for bound constrained optimization," {\it SIAM Journal of Scientific Computing}, 16, 1190-1208.

\refmark Corwin, E. \& Logar, A. (2004),
``Sorting in linear time variations on the bucket sort,''
{\it Journal of Computing Sciences in Colleges}, 20, 197-202.

\refmark Degen, M., Embrechts, P. \& Lambrigger, D. D. (2007),
``The quantitative modeling of operational risk: between g-and-h and EVT,"
{\it Astin Bulletin}, 37, 265-291.

\refmark Dutta, K. K. \& Babbel, D. F. (2002),
``On measuring skewness and kurtosis in short rate distributions:
The case of the US dollar London inter bank offer rates,''1
Technical report, The Wharton School, University of Pennsylvania.

\refmark Field, C. (2004),
``Using the $gh$ distribution to model extreme wind speeds,''
{\it Journal of Statistical Planning and Inference}, 122, 15-22.

\refmark Field, C. \& Genton, M. G. (2006),
``The multivariate $g$-and-$h$ distribution,''
{\it Technometrics}, 48, 104-111.

\refmark He, Y. \& Raghunathan, T. E. (2006),
``Tukey's gh distribution for multiple imputation,''
{\it The American Statistician}, 60, 251-256.

\refmark He, Y. \& Raghunathan, T. E. (2012),
``Multiple imputation using multivariate gh transformations,''
{\it Journal of Applied Statistics}, 39, 2177-2198.

\refmark Hoaglin, D. C. (1985),
``Summarizing shape numerically: The $g$-and-$h$ distributions,''
in {\it Data Analysis for Tables, Trends and Shapes: Robust and
Exploratory Techniques}, D. C. Hoaglin, F. Mosteller, J. W. Tukey (eds.), New York: Wiley.

\refmark  Hoaglin, D. C. (2010),
``Extreme-value distributions as $g$-and-$h$ distributions: An empirical view,''
{\it Technical report}, JSM 2010.

\refmark Hoaglin, D. C. \& Peters, S. C. (1979),
``Software for exploring distributional shapes,''
in {\it Proceedings of Computer Science and Statistics: 12th annual Symposium on the Interface},
J. F. Gentleman (ed.), University of Waterloo, Ontario, 418-443.

\refmark Jim\'{e}nez, J. A. \& Arunachalam, V. (2011),
``Using Tukey's g and h family of distributions to calculate value-at-risk and conditional value-at-risk,''
{\it Journal of Risk}, 13, 95-116.

\refmark Jones, M. C. (2015),
``On families of distributions with shape parameters (with discussions),''
{\it International Statistical Review}, to appear.

\refmark Ledolter, J. \& Hogg, R. V. (2009),
{\it Applied Statistics for Engineers and Physical Scientists}, 3rd Ed., Pearson/Prentice Hall.

\refmark MacGillivray, H. L. (1992),
``Shape properties of the g-and-h and Johnson families,''
{\it Communications in Statistics - Theory and Methods}, 21, 1233-1250.

\refmark Martinez, J. \& Iglewicz, B. (1984),
``Some properties of the Tukey $g$ and $h$ family of distributions,''
{\it Communications in Statistics, Theory and Methods}, 13, 353-369.

\refmark Morgenthaler, S. \& Tukey, J. W. (2000),
``Fitting quantiles: doubling, {\it HR}, {\it HQ}, and {\it HHH} distributions,''
{\it Journal of Computational and Graphical Statistics}, 9, 180-195.

\refmark R Development Core Team (2014), R: A language and environment for statistical computing.
R Foundation for Statistical Computing, Vienna, Austria. ISBN 3-900051-07-0, URL
http://www.R-project.org.

\refmark Rayner, G. D. \& MacGillivray, H. L. (2002a),
``Numerical maximum likelihood estimation for the g-and-k and generalized g-and-h distributions,''
{\it Statistics and Computing}, 12, 57-75.

\refmark Rayner, G. D. \& MacGillivray, H. L. (2002b),
``Weighted quantile-based estimation for a class of transformation distributions,''
{\it Computational Statistics and Data Analysis}, 39, 401-433.

\refmark Self, S. G. \& Liang, K. Y. (1987),
``Asymptotic properties of maximum likelihood estimators and likelihood ratio tests under nonstandard conditions,''
{\it Journal of the American Statistical Association}, 82, 605-610.

\refmark Tukey, J. W. (1977), {\it Modern techniques in data analysis}, NSF-sponsored regional research conference at Southeastern
Massachusetts University, North Dartmouth, MA.

\refmark Xu, G. \& Wang, S. (2010),
``A goodness-of-fit test of logistic regression models for
case-control data with measurement errors,''
{\it Biometrika}, 98, 877-886.

\refmark Xu, G., Wang, S., Huang, J.Z. (2014), ``Focused information criterion and model averaging based on weighted composite quantile regression." {\it Scandinavian Journal of Statistics}, 41, 365-381.

\refmark Xu, Y., Iglewicz, B. \& Chervoneva, I. (2014),
``Robust estimation of the parameters of g-and-h distributions, with applications to outlier detection,''
{\it Computational Statistics and Data Analysis}, 75, 66-80.

\refmark Xu, G., Xiang, Y. B., Wang, S. \& Lin, Z. Y. (2012) ``Regularization and variable selection for infinite variance autoregressive models," {\it Journal of Statistical Planning and Inference}, 142, 2545-2553.

\refmark Zhang, B. (2001),
``An information matrix test for logistic regression models based on case-control data,''
{\it Biometrika}, 88, 921-932.
\end{document}